\documentclass[12pt]{article}

\usepackage{bbm}
\usepackage{chemarrow}
\usepackage{color}
\usepackage{fullpage}
\usepackage{libertine}

\usepackage{cite}
\RequirePackage[colorlinks,citecolor=blue,urlcolor=blue]{hyperref}

\usepackage{amssymb}
\usepackage{amsmath}
\usepackage{amsthm}
\usepackage{mathrsfs}
\usepackage[pdftex,dvipsnames]{xcolor}
\usepackage{graphicx}
\usepackage{float}
\usepackage[ruled,vlined]{algorithm2e}
\usepackage{mathrsfs}
\usepackage{algpseudocode}
\usepackage{subfigure}
\newcommand{\bs}[1]{\boldsymbol{#1}}

\usepackage{xcolor}

\newcommand{\prob}{\mathcal{P}} 
\newcommand{\argmin}{\mathbf{argmin}} 
\newcommand{\OTop}{\mathscr{T}} 

\renewcommand{\phi}{\varphi}

\renewcommand{\epsilon}{\varepsilon}
\newcommand{\cX}{\mathcal{X}}

\newcommand{\cost}{\mathbf{c}}
\newcommand{\costmat}{\mathbf{D}}
\newcommand{\mass}{\mathbf{M}}
\newcommand{\summary}{\mathcal{S}}
\newcommand{\zz}{\mathbf{z}}
\newcommand{\uu}{\mathbf{u}}
\newcommand{\forward}{\mathcal{F}}

\newcommand{\C}{\mathbf{C}}

\newcommand{\E}[1]{\mathbb{E}\left[ #1  \right]}

\newcommand{\wc}{\xrightarrow{\text{w}}}
\newcommand{\vertiii}[1]{{\left\vert\kern-0.25ex\left\vert\kern-0.25ex\left\vert #1 
    \right\vert\kern-0.25ex\right\vert\kern-0.25ex\right\vert}}

\newcommand{\bra}[1]{\langle #1 \rangle}
\newcommand{\BK}[1]{ {\left( #1 \right)} }
\newcommand{\sqBK}[1]{ {\left[ #1 \right]} }
\newcommand{\curBK}[1]{ {\left\{ #1 \right\}} }

\newcommand{\ess}{\ensuremath{\mathrm{ESS}}}
\newcommand{\noise}{\ensuremath{\mathrm{noise}}}
\newcommand{\IS}{\ensuremath{\mathrm{IS}}}
\newcommand{\OT}{\ensuremath{\mathrm{OT}}}

\newcommand{\prior}{\textrm{prior}}
\newcommand{\post}{\textrm{post}}

\newcommand{\mesh}{\mathcal{M}}
\newcommand{\hess}{\mathcal{H}}

\newcommand{\bayesop}{\mathscr{B}}
\newcommand{\resampop}{\mathscr{R}}
\newcommand{\markovop}{\mathscr{M}}

\newtheorem{thm}{Theorem}
\newtheorem{theorem}[thm]{Theorem}
\newtheorem{prop}[thm]{Proposition}

\newtheorem{assumption}[thm]{Assumption}
\newtheorem{remark}[thm]{Remark}

\newtheorem{lemma}[thm]{Lemma}

\begin{document}

\title{Sequential Ensemble Transform for Bayesian Inverse Problems}

\author{Aaron Myers \footnotemark[1]  
  , Alexandre H. Thiery \footnotemark[4]
  , Kainan Wang  \footnotemark[2] 
  , Tan Bui-Thanh \footnotemark[1] \footnotemark[3]}

\renewcommand{\thefootnote}{\fnsymbol{footnote}}

\footnotetext[1]{Institute for Computational Engineering
    \& Sciences, The University of Texas at Austin, Austin, TX 78712,
    USA.}

\footnotetext[2]{Sanchez Oil \& Gas, Houston, TX}

\footnotetext[3]{Department of Aerospace Engineering \& Engineering Mechanics
	, The University of Texas at Austin, Austin, TX 78712, USA}

\footnotetext[4]{Department of Satistics and Applied Probability,
	National University of Singapore, Singapore }

\renewcommand{\thefootnote}{\arabic{footnote}}

\newcommand{\TODO}[1]{ \fbox{\parbox{3in}{\bf TODO: #1}}}

\newcommand{\grbf}[1] {\mbox{\boldmath${#1}$\unboldmath}}
\newcommand{\gbf}[1] {\mathbf{#1}}

\newcommand{\beq} {\begin{equation}}
\newcommand{\eeq} {\end{equation}}
\newcommand{\bdm} {\begin{displaymath}}
\newcommand{\edm} {\end{displaymath}}
\newcommand{\bit}{\begin{itemize}}
\newcommand{\eit}{\end{itemize}}
\newcommand{\bde}{\begin{description}}
\newcommand{\ede}{\end{description}}
\newcommand{\bce}{\begin{center}}
\newcommand{\ece}{\end{center}}
\newcommand{\ben} {\begin{enumerate}}
\newcommand{\een} {\end{enumerate}}
\newcommand{\bea} {\begin{eqnarray}}
\newcommand{\eea} {\end{eqnarray}}
\newcommand{\barr} {\begin{array}}
\newcommand{\earr} {\end{array}}
\newcommand{\bean} {\begin{eqnarray*}}
\newcommand{\eean} {\end{eqnarray*}}
\newcommand{\edoc} {\end{document}}
\newcommand{\On}{\hspace{0.2cm} \mbox{on} \hspace{0.2cm}}
\newcommand{\In}{\hspace{0.2cm} \mbox{in} \hspace{0.2cm}}
\newcommand{\Minimize}{\hspace{0.2cm} \mbox{minimize} \hspace{0.2cm}}
\newcommand{\Maximize}{\hspace{0.2cm} \mbox{maximize} \hspace{0.2cm}}
\newcommand{\SubjectTo}{\hspace{0.2cm} \mbox{subject} \hspace{0.15cm}
            \mbox{to} \hspace{0.2cm}}

\def\etal{{\it et al.~}}

\newcommand{\point}[1] {\ensuremath{\boldsymbol{#1}}}
\newcommand{\vvec}[1] {\ensuremath{\boldsymbol{#1}}}
\renewcommand{\vec}[1] {\ensuremath{\boldsymbol{#1}}}
\newcommand{\cvec}[1] {\ensuremath{\boldsymbol{{#1}}}}
\newcommand{\ten}[1] {\ensuremath{\boldsymbol{#1}}}
\newcommand{\tenfour}[1] {\ensuremath{\boldsymbol{\mathsf{#1}}}}
\newcommand{\comp}[1]{\begin{pmatrix}{ #1 }\end{pmatrix}}
\newcommand{\vv}{\ensuremath{\vec{v}}}
\newcommand{\vr}{\left(\ensuremath{\vec{r}}\right)}
\newcommand{\att}{\left(t\right)}
\newcommand{\tvv}{\ensuremath{\tilde{\vec{v}}}}
\newcommand{\ww}{\ensuremath{\vec{w}}}
\newcommand{\GG}{\ensuremath{\ten{G}}}
\newcommand{\HH}{\ensuremath{\ten{H}}}
\newcommand{\EE}{\ensuremath{\vec{E}}}
\newcommand{\FF}{\ensuremath{\boldsymbol{\mathfrak{F}}}}
\newcommand{\tEE}{\ensuremath{\tilde{\vec{E}}}}
\newcommand{\nn}{\ensuremath{\vec{n}}}
\renewcommand{\SS}{\ensuremath{\ten{S}}}
\newcommand{\tSS}{\ensuremath{\tilde{\ten{S}}}}
\newcommand{\cp}{\ensuremath{{c_p}}}
\newcommand{\cs}{\ensuremath{{c_s}}}
\newcommand{\tcs}{\ensuremath{{\tilde{c}_s}}}
\newcommand{\Vp}{\ensuremath{{V^e_p}}}
\newcommand{\Vs}{\ensuremath{{V^e_s}}}
\newcommand{\Sp}{\ensuremath{{S^e_p}}}
\newcommand{\Ss}{\ensuremath{{S^e_s}}}
\newcommand{\nxnx}[1] {\ensuremath{\nn\times\left(\nn\times{#1}\right)}}

\newcommand{\dd}[2] {\ensuremath{\frac{\partial {#1}}{\partial {#2}}}}
\newcommand{\DD}[2] {\ensuremath{\frac{d {#1}}{d {#2}}}}
\newcommand{\Grad} {\ensuremath{\nabla}}
\newcommand{\Gradv}[1]{\ensuremath{\nabla_{#1}}}
\newcommand{\Div} {\ensuremath{\nabla\cdot}}
\newcommand{\Divv}[1] {\ensuremath{\nabla_{#1}\cdot}}
\newcommand{\Curl} {\ensuremath{\nabla\times}}
\newcommand{\Cten}{\ensuremath{\tenfour{C}}}

\newcommand{\eval}[2][\right]{\relax
  \ifx#1\right\relax \left.\fi#2#1\rvert}

\providecommand{\abs}[1]{\left\lvert#1\right\rvert}
\providecommand{\norm}[1]{\left\lVert#1\right\rVert}

\newcommand{\Nel} {\ensuremath{{N_\text{el}}}}
\newcommand{\Ned} {\ensuremath{{N_\text{ed}}}}
\newcommand{\De} {\ensuremath{{\mathsf{D}^e}}}
\newcommand{\Dep} {\ensuremath{{\mathsf{D}^{e'}}}}
\newcommand{\Dhat} {\ensuremath{\hat{\mathsf{D}}}}
\newcommand{\Dehat} {\ensuremath{{\hat{\mathsf{D}}^e}}}
\newcommand{\half} {\ensuremath{\frac{1}{2}}}
\newcommand{\IN} {\ensuremath{{\mathcal{I}_N}}}
\newcommand{\INe} {\ensuremath{{\mathcal{I}_{N_e}}}}

\newcommand{\mc}[1]{\mathcal{#1}}
\newcommand{\mcC}[1]{\mathcal{#1}}
\newcommand{\mb}[1]{\mathbf{#1}}
\newcommand{\mbb}[1]{\mathbb{#1}}
\newcommand{\mi}[1]{\mathit{#1}}
\newcommand{\mf}[1]{\mathfrak{#1}}
\newcommand{\Oi}[1]{\int_{\mc{D}} #1 \, d\vec{x}}
\newcommand{\TOi}[1]{\int_{0}^T \int_{\mc{D}} #1 \, d\vec{x}dt}
\newcommand{\Ti}[1]{\int_{0}^T #1 \, dt}
\newcommand{\TDei}[1]{\int_{0}^T\int_{\De} #1 \, d\vec{x}dt}
\newcommand{\TDeid}[1]{\int_{0}^T\int_{\De,N} #1 \, d\vec{x}dt}
\newcommand{\Dei}[1]{\int_{\De} #1 \, d\vec{x}}
\newcommand{\DeI}[1]{\int_{D} #1 \, d\vec{x}}
\newcommand{\DeiM}[1]{\int_{\Dhat} #1 \, d\vec{r}}
\newcommand{\DeMd}[1]{\int_{\De,N_e} #1 \, d\vec{x}}
\newcommand{\DeiMd}[1]{\int_{\Dhat,N_e} #1 \, d\vec{r}}
\newcommand{\DeMdN}[1]{\int_{D,N} #1 \, d\vec{x}}
\newcommand{\DeiMdN}[1]{\int_{\Dhat,N} #1 \, d\vec{r}}
\newcommand{\DeiMdNC}[2]{\int_{\Dhat,N_{#2}} #1 \, d\vec{r}}
\newcommand{\pDei}[1]{\int_{\partial \De} #1 \, d\vec{x}}
\newcommand{\pDeiM}[1]{\int_{\partial \Dhat} #1 \, d\vec{r}}
\newcommand{\pDeiMd}[1]{\int_{\partial \Dhat,N_e} #1 \, d\vec{r}}
\newcommand{\pDeiMdN}[1]{\int_{\partial \Dhat,N} #1 \, d\vec{r}}
\newcommand{\pDeiMdNC}[2]{\int_{\partial \Dhat,N_{#2}} #1 \, d\vec{r}}
\newcommand{\pMortar}[2]{\int_{\mc{M}_{#2}} #1 \, d\vec{r}}
\newcommand{\pMortard}[2]{\int_{\mc{M}^{#2},N_{#2}} #1 \, d\vec{r}}
\newcommand{\pMortardd}[3]{\int_{\mc{M}^{#2},N_{#3}} #1 \, d\vec{r}}
\newcommand{\pMortardh}[3]{\int_{\mc{M}_{#2},N_{#3}} #1 \, d\vec{r}}
\newcommand{\pMortardhn}[4]{\int_{{#2}_{#3},N_{#4}} #1 \, d\vec{r}}
\newcommand{\TpDei}[1]{\int_{0}^T\int_{\partial \De} #1 \, d\vec{x}}
\newcommand{\TpDeid}[1]{\int_{0}^T\int_{\partial \De,N} #1 \, d\vec{x}}
\newcommand{\Deid}[1]{\int_{\De,N_e} #1 \, d\vec{x}}
\newcommand{\DeidN}[1]{\int_{\De,N} #1 \, d\vec{x}}
\newcommand{\pDeidN}[1]{\int_{\partial \De,N} #1 \, d\vec{x}}
\newcommand{\pDeid}[1]{\int_{\partial \De,N_e} #1 \, d\vec{x}}
\newcommand{\nor}[1]{\left\| #1 \right\|}
\newcommand{\snor}[1]{\left| #1 \right|}
\newcommand{\LRp}[1]{\left( #1 \right)}
\newcommand{\LRs}[1]{\left[ #1 \right]}
\newcommand{\LRa}[1]{\left< #1 \right>}
\newcommand{\LRc}[1]{\left\{ #1 \right\}}
\newcommand{\pp}[2]{\frac{\partial #1}{\partial #2}}
\newcommand{\ppcp}[1]{\frac{\partial #1}{\partial \vec{c_p}}}
\newcommand{\ppcpj}[1]{\frac{\partial #1}{\partial c_p}}
\newcommand{\ppcpmf}[1]{\frac{\partial #1}{\partial \vec{\mf{c_p}}}}
\newcommand{\ppm}[1]{\frac{\partial #1}{\partial \vec{m}}}
\newcommand{\ppho}[3]{\frac{\partial^{#3} #1}{{\partial #2}^{#3}}}
\newcommand{\ppmi}[1]{\frac{\partial #1}{\partial m_i}}
\newcommand{\ppmj}[1]{\frac{\partial #1}{\partial m_j}}
\newcommand{\ppcpm}[1]{\frac{\partial #1}{\partial \vec{c_p}^-}}
\newcommand{\ppcpp}[1]{\frac{\partial #1}{\partial \vec{c_p}^+}}
\newcommand{\ppcs}[1]{\frac{\partial #1}{\partial \vec{c_s}}}
\newcommand{\ppcsm}[1]{\frac{\partial #1}{\partial \vec{c_s}^-}}
\newcommand{\ppcsp}[1]{\frac{\partial #1}{\partial \vec{c_s}^+}}
\newcommand{\td}[2]{\frac{d #1}{d #2}}
\newcommand{\Rvert}[1]{\left. #1 \right|}
\newcommand{\ipDed}[3]{\LRp{ #1, #2}_{\De,N}^{#3}}

\newcommand{\bigO}{\mathcal{O}}

\newcommand{\iOm}[1]{\int_{\Omega} #1 \, d\Omega}
\newcommand{\iGb}[1]{\int_{\Gamma_{b}} #1 \, ds}
\newcommand{\iGs}[1]{\int_{\Gamma_{s}} #1 \, ds}
\newcommand{\iGsy}[1]{\int_{\Gamma_{s}} #1 \, ds(\mb{y})}
\newcommand{\RE}{\mbox{{I}}\hspace{-0.5ex}\mbox{{R}}}
\newcommand{\iC}[1]{\int_{0}^{2\pi} #1 \, d\theta}
\newcommand{\iGr}[1]{\int_{\Gamma_{\infty}} #1 \, ds}

\newcommand{\jump}[1] {\ensuremath{\LRs{\![#1]\!}}}
\newcommand{\diff}[1] {\ensuremath{\LRs{#1}}}
\newcommand{\average}[1] {\ensuremath{\LRc{\!\{#1\}\!}}}

\newcounter{tablecell}

\newcommand*{\numbercell}{%
  \stepcounter{tablecell}%
  \thetablecell. 
}

\newtheorem{propo}{Proposition}
\newtheorem{coro}{Corollary}
\newtheorem{theo}{Theorem}
\newtheorem{lem}{Lemma}

\newtheorem{defi}{definition}

\newtheorem{rema}{remark}

\newcommand{\figlab}[1]{\label{fig:#1}}
\newcommand{\eqnlab}[1]{\label{eq:#1}}
\newcommand{\theolab}[1]{\label{theo:#1}}
\newcommand{\corolab}[1]{\label{coro:#1}}
\newcommand{\propolab}[1]{\label{propo:#1}}
\newcommand{\lemlab}[1]{\label{lem:#1}}
\newcommand{\defilab}[1]{\label{defi:#1}}
\newcommand{\remalab}[1]{\label{rema:#1}}
\newcommand{\tablab}[1]{\label{tab:#1}}

\newcommand{\figref}[1]{\ref{fig:#1}}
\newcommand{\theoref}[1]{\ref{theo:#1}}
\newcommand{\defiref}[1]{\ref{defi:#1}}
\newcommand{\remaref}[1]{\ref{rema:#1}}
\newcommand{\cororef}[1]{\ref{coro:#1}}
\newcommand{\proporef}[1]{\ref{propo:#1}}
\newcommand{\lemref}[1]{\ref{lem:#1}}
\newcommand{\eqnref}[1]{\eqref{eq:#1}}
\newcommand{\alglab}[1]{\label{alg:#1}}
\newcommand{\seclab}[1]{\label{sect:#1}}
\newcommand{\secref}[1]{\ref{sect:#1}}
\newcommand{\tabref}[1]{\ref{tab:#1}}

\newcommand{\GM}[2]{\mc{N}\left( #1, #2 \right)}

\newcommand{\w}{w}
\newcommand{\W}{W}
\renewcommand{\b}{b}
\renewcommand{\u}{u}
\newcommand{\ub}{\bs{\u}}
\renewcommand{\v}{v}
\newcommand{\f}{f}
\newcommand{\g}{g}
\newcommand{\pOmega}{\partial\Omega}
\renewcommand{\d}{d}
\newcommand{\db}{\mb{d}}
\newcommand{\etab}{\bs{\eta}}
\newcommand{\Ltwo}{L^2\LRp{\Omega}}
\newcommand{\X}{X}
\renewcommand{\L}{{\mb{L}}}
\newcommand{\R}{\mathbb{R}}
\newcommand{\A}{A}
\newcommand{\B}{B}
\newcommand{\M}{M}
\newcommand{\K}{K}
\newcommand{\q}{q}
\newcommand{\V}{V}
\newcommand{\zb}{\mb{z}}
\renewcommand{\sb}{\mb{s}}
\newcommand{\rb}{\mb{r}}

\maketitle

\begin{abstract}
We present the Sequential Ensemble Transform (SET) method, an approach for generating approximate samples from a Bayesian posterior distribution. The method explores the posterior distribution by solving a sequence of discrete optimal transport problems to produce a series of transport plans which map prior samples to posterior samples.
We prove that the sequence of Dirac mixture distributions produced by the SET method converges weakly to the true posterior as the sample size approaches infinity. Furthermore, our numerical results indicate that, when compared to  standard Sequential Monte Carlo (SMC) methods, the SET approach is more robust to the choice of Markov mutation kernels and requires less computational efforts to reach a similar accuracy when used to explore complex posterior distributions. Finally, we describe adaptive schemes that allow to completely automate the use of the SET method.
\end{abstract}

\section{Introduction}
\label{sec:intro}
Inverse problems  enable 
integration of observational and experimental data, simulations and/or
mathematical models to make scientific predictions. 
We focus on inverse problems in which the goal is to
determine a parameter of interest from indirect
and imprecise observations. 
The relationship between the parameter and the
noise-free observations, the forward map, is often provided through the solution of a complex mathematical
model, the forward problem.

The Bayesian approach 
formulates the inverse problem as a statistical inference
problem \cite{mosegaard1995monte,Stuart2010,Kaipio2006}. Given noisy
observational data, the governing forward
problem, and a prior probability distribution,
the solution of the Bayesian inverse
problem is the posterior probability distribution over the
parameters.  The prior distribution encodes knowledge or assumptions about the
parameter space before data are observed. 
The posterior distribution incorporates both the prior knowledge and the
observations. Non-linearity of the forward map leads to posterior distributions that 
are typically not Gaussian, even in situations when both the prior and observational noise 
probability distributions
are Gaussian.

Exploring a high dimensional non-Gaussian posterior
is computationally challenging. Indeed, evaluating the posterior density 
typically requires evaluating the forward map which, for problems governed by partial differential equations (PDEs), 
dominates the computational cost. Standard numerical quadrature methods routinely used for estimating 
statistical quantities of interest (e.g. statistical moments, probability of rare event) are  
infeasible in these high-dimensional settings.

The Markov chain Monte Carlo (MCMC) algorithm \cite{Hastings70,
  MetropolisRosenbluthRosenbluthEtAl53} 
is a popular approach for exploring the posterior distribution in Bayesian inverse problems. 
Estimates obtained from standard MCMC methods often require a large number of samples to be meaningful, especially in high dimensional settings. In Bayesian inverse problems, generating each MCMC sample requires 
an evaluation of the posterior density, which relies on evaluating the computationally expensive forward map.

Sequential Monte Carlo (SMC) methods are computational
techniques widely used in engineering, statistics, and many other
fields \cite{gordon1993novel,doucet2001introduction,DelMoral2004,doucet2009tutorial,DelDoucetJasra06} to
approximate a sequence of probability distributions, usually of
increasing complexity or dimension.  
A standard approach in Bayesian inverse problems consists of introducing a sequence of distributions that interpolates between a distribution that is easy to sample from (e.g. the prior distribution, or a Gaussian approximation of the posterior distribution) and the posterior distribution. Through a combination of importance sampling, 
Markovian mutations and resampling procedures, the SMC method iteratively 
constructs a sequence of particle approximations of this sequence of 
distributions. Under very mild assumptions, 
SMC methods are consistent in the limit when the number $N$ of particles 
goes to infinity and converge at Monte-Carlo rate $\mathcal{O}(N^{-1/2})$. 
Furthermore, methods are available for implementing this class of algorithms 
on parallel architectures \cite{whiteley2016role,verge2015parallel,lee2016forest,sen2019particle}.

In this article, inspired by recent developments in the data-assimilation 
literature \cite{Reich2013,Cheng2013}, we exploit algorithms based on the concept of optimal transport \cite{Monge1781,Villani2008,Villani2003,peyre2019computational}. 
Our approach, 
the Sequential Ensemble Transform (SET) method, combines the SMC 
framework with the use of optimal transport to efficiently build 
particle approximations of the posterior distribution in high-dimensional 
Bayesian inverse problems (see figure \ref{fig:smcex}). We refer the readers  
to \cite{MoselhyMarzouk12,heng2015gibbs,Parno2014,spantini2018inference} for other Monte-Carlo methods 
based on transportation concepts.
Unlike SMC methods,  the SET approach, similarly  
to the algorithm of \cite{Reich2013}, uses an optimal 
transport scheme instead of the usual resampling procedure. The main advantage 
of the proposed method is its robustness with respect to the choice of mutation kernel steps. 
Indeed, without mutation kernel, the SMC method is a  variant of the
standard importance sampling procedure, which is known to behave poorly in 
high-dimensional settings \cite{bengtsson2008curse}, or more generally when there is a large discrepancy 
between the proposal and target distributions. 
Consequently, good mutation kernels are often crucial 
to the successful implementation of SMC methods in Bayesian inverse problems \cite{Beskos2014}. 
Unfortunately, it is notoriously difficult to design Markov mutation kernels with good mixing properties 
in  high-dimensional settings that are common in Bayesian inverse problems \cite{Bui-ThanhGhattasMartinEtAl13, Beskos2015,Kantas2014}.
Adaptive SMC procedures \cite{chopin2002sequential,del2012adaptive,jasra2011inference,Beskos2015a} can help mitigate this issue by 
automatically tuning the mutation kernels and the interpolating sequence of distributions. 
Our numerical studies presented in Section \ref{sec:numerical} show that the SET 
approach performs favorably when compared to standard SMC methods. Furthermore, although approximate methods \cite{genevay2016stochastic,Cuturi13} are available for efficiently solving discrete optimal transport problems, we have found that in most realistic Bayesian inverse problems and for a typical number of particles $N \lesssim 10^4$, the computational cost of (exactly) solving the discrete optimal transport problems is negligible when compared to the computational burden associated with the forward-solves necessary to implement the SET/SMC algorithms.
Finally, it should be mentioned that in situations (such as low
dimensional parameter spaces or closed-to-Gaussian posteriors) when
the design of Markov kernels with good mixing properties is not
challenging, our proposed method may not provide significant computational savings 
over more standard SMC or MCMC methods. 

Our main contributions are as follows. We propose the Sequential Ensemble Transform (SET) algorithm, an interacting particle methodology inspired from the data-assimilation literature \cite{Reich2013}, for Bayesian inversion.
Unlike most interacting particle methods that rely on resampling
approaches, the SET method is based on optimal transportation. We
demonstrate empirically that this leads to an algorithm that is less 
affected by particle degeneracy, and requires less computational effort to converge, 
than more standard SMC approaches when used
in complex settings where designing efficient Markov mutation kernels
is not trivial. 
We make SET practical, especially for complex applications,  by 
providing
several adaptation strategies for automating the choice of tuning parameters.
Finally, we establish 
conditions under which, in the limit
when the number of particles approaches infinity, 
the SET method is provably consistent,
i.e., the sequence of particle approximations produced by the SET 
converges weakly towards the underlying target distribution.

\begin{figure}[h!t!b!]
\begin{center}
\subfigure{
  \includegraphics[trim={0cm 0cm 0cm 0cm},clip,width=0.30\columnwidth]{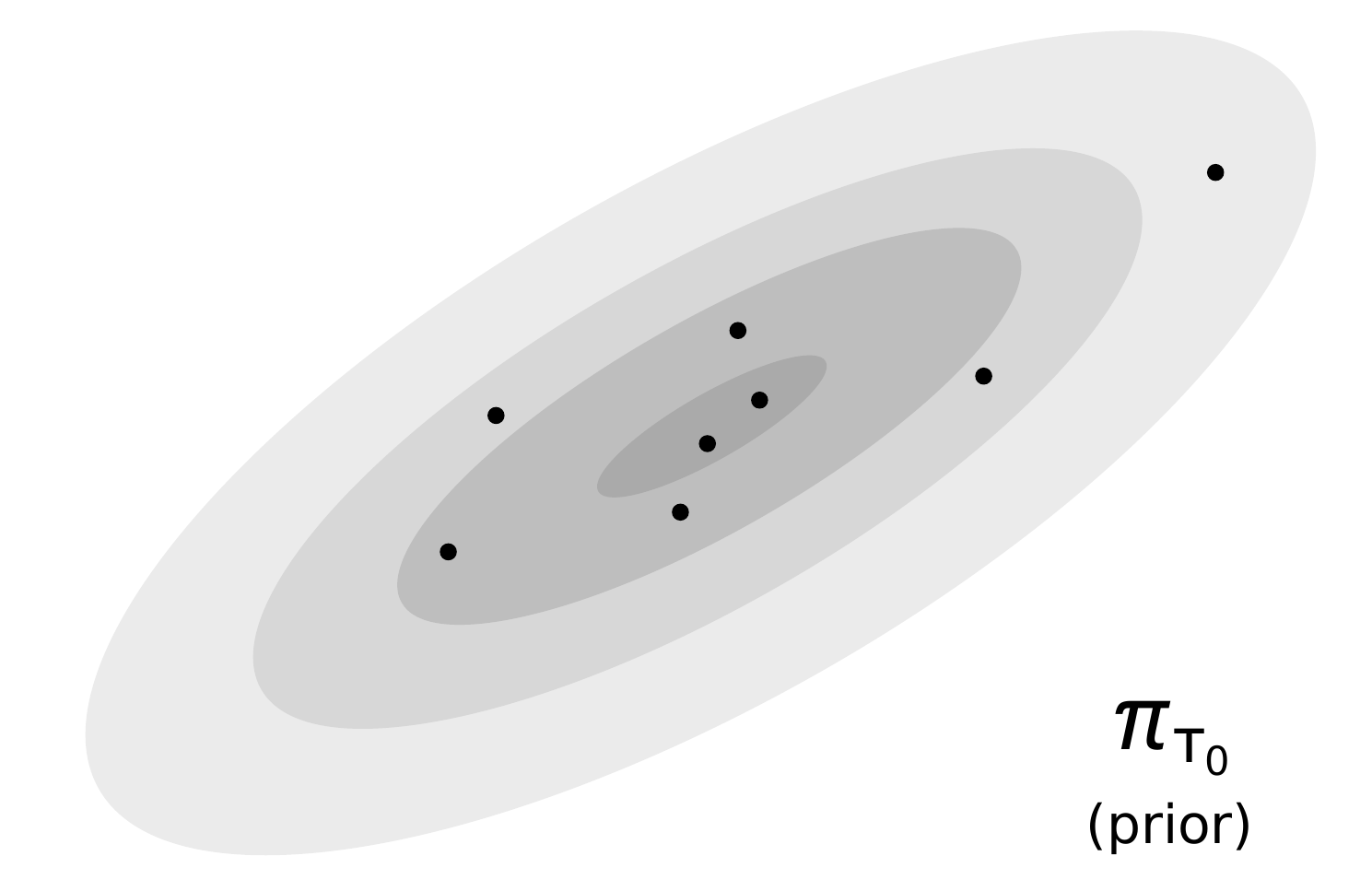}
  }
\subfigure{
  \includegraphics[trim={0cm 0cm 0cm 0cm},clip,width=0.30\columnwidth]{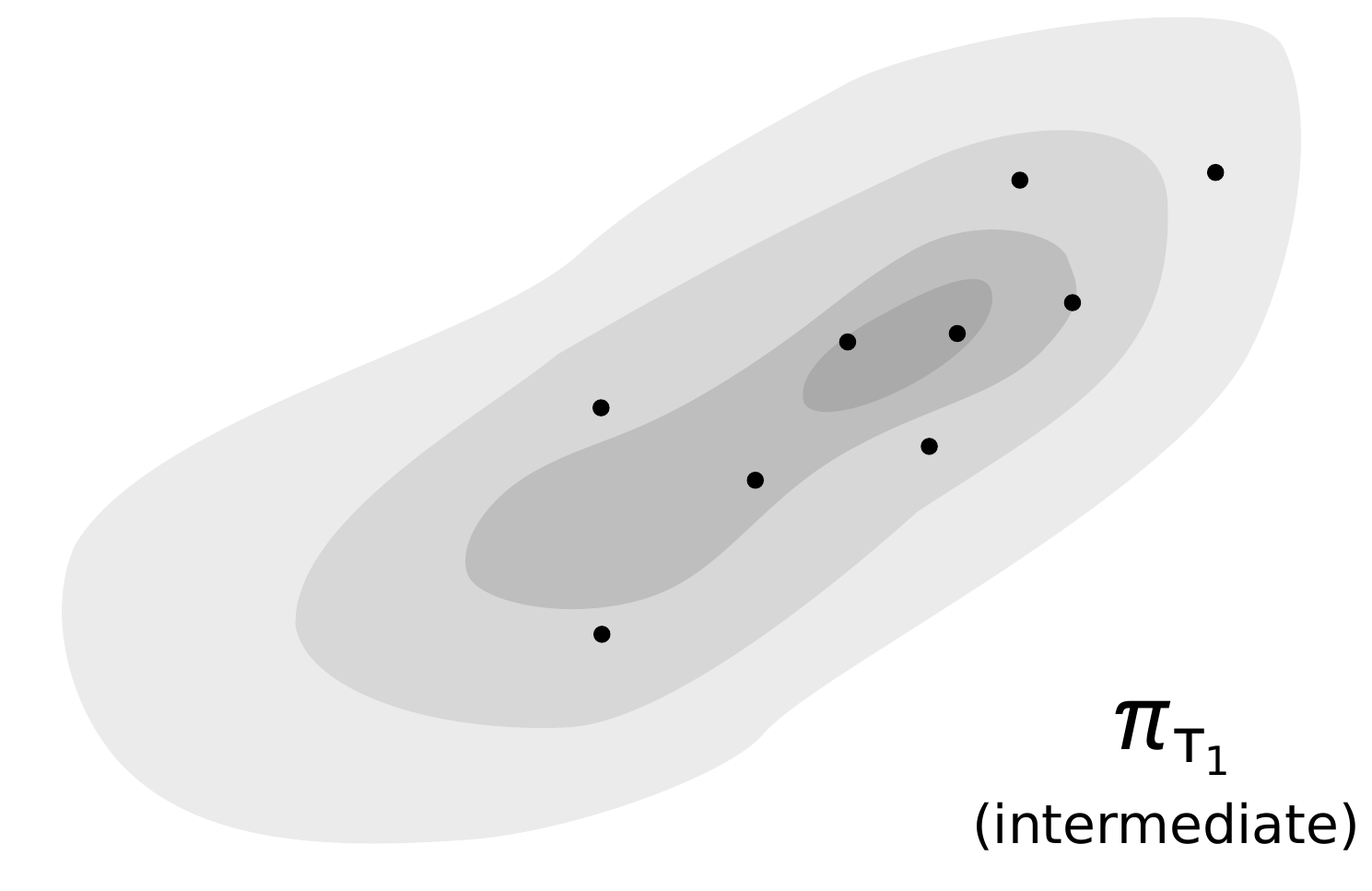}
  }
\subfigure{
  \includegraphics[trim={0cm 0cm 0cm 0cm},clip,width=0.30\columnwidth]{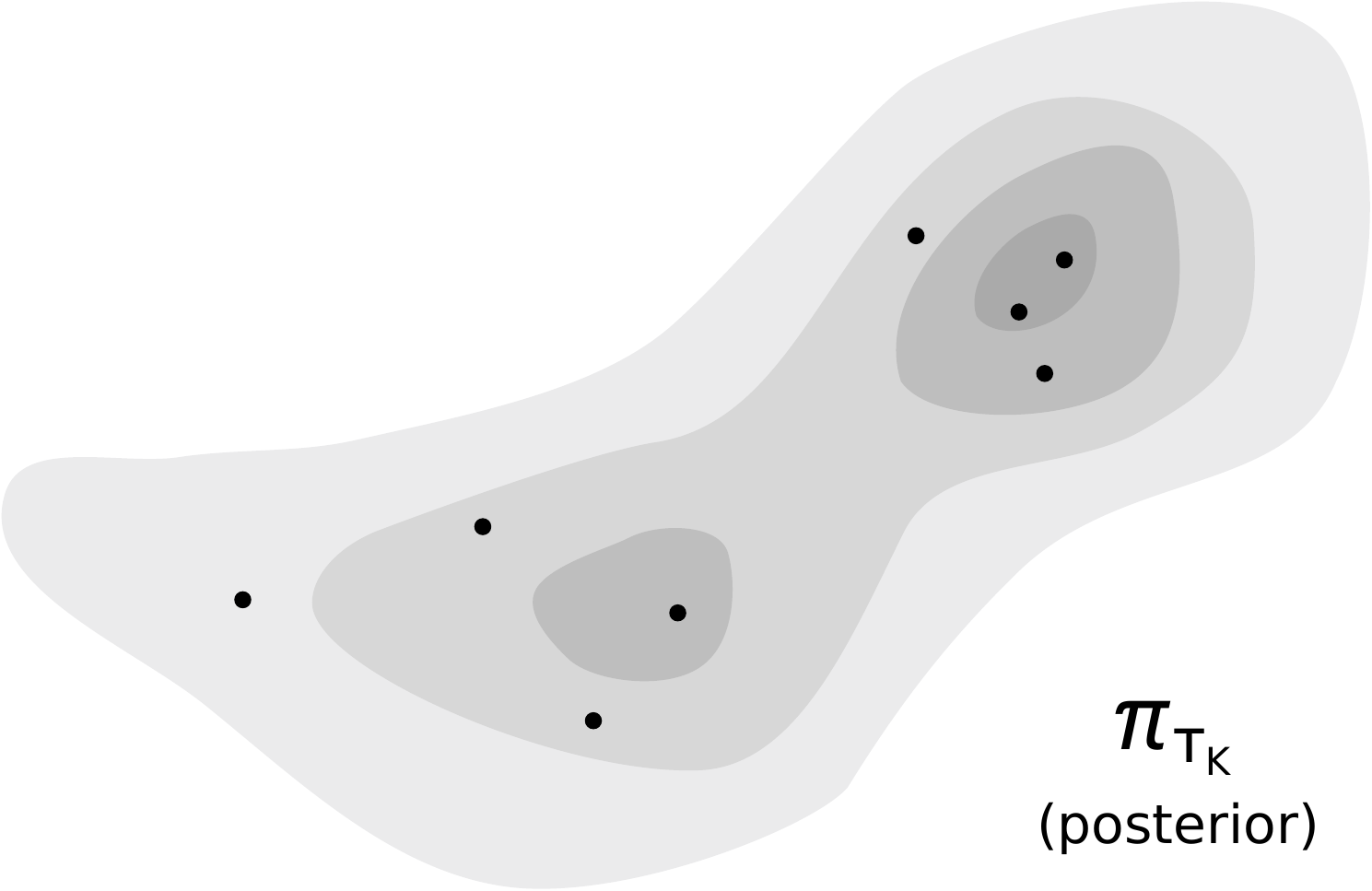}
  }

 \end{center}
 \caption{A representation of the SET method using optimal transport to move particles in parameter space as to 
 represent the posterior}
\label{fig:smcex}
\end{figure}

The article is structured as follows. In Section
\ref{sec:problemStatement}, PDE-constrained Bayesian inverse problems are briefly described. An overview of particle methods and importance sampling is presented in Section \ref{sec:particlemethods}.
Section \ref{subsec:optimaltrans} presents the concept of optimal transport 
and describes the main components of SET method, as well as their asymptotic properties. 
Section \ref{subsec.SET} describes the SET methods in details, as well as several adaptive strategies that can be used to automate several aspects of the method.
Finally, Section \ref{sec:numerical} presents various numerical
results, including a Bayesian inverse problem with a non-linear forward map. 
Section \ref{sec:conclusions} concludes the paper and discusses future work.

\subsection*{Notations and conventions}
Unless stated otherwise, all the state spaces are endowed with a
metric and the associated Borel $\sigma$-algebra. 
The notations $\mu$ and $\nu$ (along with any use of super- or sub-scripts) denote probability distributions.
A sequence of probability distributions
$\{ \mu^N \}_{N \geq 1}$ on $\mathcal{X}$ converges weakly towards the
distribution $\mu$, denoted as $\mu^N \wc \mu$, if for any bounded and continuous 
test function $\phi:\mathcal{X} \to \mathbb{R}$ we have 
that $\int \phi(u) \, \mu^N(du) \to \int \phi(u) \, \mu(du)$ as $N \to \infty$. 
Similarly, a sequence of random probability distribution $\mu_\omega^N$ almost surely 
converges weakly towards $\mu$ if, for $\mathbb{P}$-almost every $\omega$, we have that $\mu^N_\omega \wc \mu$.
The set of probability distributions on a state space $\mathcal{X}$ is
denoted as $\prob(\mathcal{X})$. For a set $S$, the notation $\mathbbm{1}_S$ 
refers to the indicator function of $S$, i.e., the function that equals one for $x \in S$ 
and zero otherwise. For $u \in \mathcal{X}$, the Dirac probability 
distribution $\delta(u)$ is the distribution 
that puts all its probability mass at $u$.


\section{Problem Statement}
\label{sec:problemStatement}
Although the methods described in this article are general, for
illustration purposes, we focus on the Bayesian treatment of inverse problems. 
We are interested in estimating a field $u \in \mathcal{X}$, 
where $\mathcal{X}$ denotes a space of functions, 
from a finite set of observations contaminated by additive Gaussian noise,
\begin{align*}
\db =\mc{G}(u) + \etab,
\end{align*}
where $\db = [d_1, \ldots, d_D]^\top \in \mathcal{Y}$ and $\etab \sim \GM{0}{\L}$ is centred Gaussian vector 
with covariance matrix $\L$. 
The operator $\mc{G}:\mathcal{X} \to \mathcal{Y}$ describes the mapping from the parameters to observables. In Bayesian inverse problems and as illustrated in Section \ref{sec:numerical}, estimating the quantity $\mc{G}(u)$ typically involves solving a set of partial differential equations. In order to estimate the uncertainty associated to the necessarily imperfect reconstruction of the parameter $u \in \mathcal{X}$, the Bayesian approach postulates a prior distribution $\mu_{\prior}$ that describes the information available on the parameter $u$ prior to any data collection. Under mild assumptions \cite{Stuart2010}, the Bayesian posterior distribution $\mu_{\post}$ is defined through the change of measure formula
\begin{equation}
\label{eq.posterior}
\frac{d\mu_{\post}}{d\mu_{\prior}}\LRp{\u} 
\propto
\exp\curBK{-\half\snor{\db -\mc{G}\LRp{\u}}^2_{\L}}
\end{equation}
where $\snor{\cdot}_{\L}\equiv \snor{{\L}^{-\half}\cdot}$ denotes the ${\L}^{-\half}$-weighted 
Euclidean norm. In situations when the mapping $\mc{G}$ is non-linear, the posterior 
distribution is typically intractable and numerical methods such as MCMC 
are required to estimate expectations (and other statistics) of observables with respect to the posterior $\mu_{\post}$.

\section{Particle Methods}
\label{sec:particlemethods}
Particle methods approximate probability distributions with weighted mixtures of Diracs, also referred to as {\it particle approximation} in this text.
To construct a particle approximation of the posterior distribution, the SMC and SET approaches proceed by introducing a sequence $\{\mu_k\}_{k=0}^K$ of distributions that interpolates between a distribution that is easy to sample from, i.e. $\mu_0$, and the posterior distribution $\mu_K$. A standard choice for $\mu_0$ is the prior distribution, or a Gaussian approximation of the posterior distribution obtained through efficient deterministic methods. For any index $1 \leq k \leq K$, set
\begin{align} \label{eq.bridge}
\frac{d \mu_k}{d \mu_{k-1}}(u) = \frac{1}{Z_k} \, \Psi_k(u),
\end{align}
for a $\mu_{k-1}$-integrable potential function $\Psi_k: \cX \to (0,\infty)$ and (typically unknown) normalization constant $Z_k > 0$. The SMC algorithm recursively 
constructs particle approximations 
$$\mu^N_k  = \frac{1}{N} \sum_{i=1}^N \, \delta(u^N_{k,i}) \, \approx \mu_k,$$
where $N \geq 1$ denotes the number of particles, by iterating {\it re-weighting}, {\it resampling}, and {\it mutation} operations that are described below. In the remaining of this text, we make use of the following notations that are standard in the Monte-Carlo literature and compactly allow to describe expectations with respect to probability distributions and Markov kernels. For a probability distribution $\mu$ on the state space $\mathcal{X}$ and a $\mu$-integrable test function $\phi:\mathcal{X} \to \mathbb{R}$, set $\mu(\phi) \equiv \int \phi(u) \, \mu(du)$. Similarly, for a Markov kernel $M(u, dv)$, define $(M \phi)(u) \equiv \int \, \phi(v) \, M(u, dv)$. 

\subsection{Re-weighting}
Consider two probability distributions $\mu$ and $\nu$ defined on the same state space $\cX$ 
and related by a change of measure 
(Radon-Nikodym derivative)
\begin{align} \label{eq.chg.proba}
\frac{d \nu}{d \mu}(u) = \frac{1}{Z} \, \Psi(u)
\end{align}
for a $\mu$-integrable potential function $\Psi: \cX \to (0,\infty)$ and a possibly unknown normalization constant $Z > 0$.
Suppose that, for any integer $N \geq 1$, it is possible to generate a set of $N$ 
particles $\{u^N_i\}_{i=1}^{N} \subset \cX$ such that the sequence of equally 
weighted particle approximations, 
\begin{align*}
\mu^N \equiv \frac{1}{N} \, \sum\limits_{i=1}^N \delta(u^N_i),
\end{align*}
converges weakly towards $\mu$ as $N \to \infty$.
Under mild assumptions, the sequence of {\it self-normalized importance sampling} weighted particle approximations $\nu^N$ defined as 
\begin{align} \label{eq.IS}
\nu^N \equiv \sum_{i=1}^N w^N_i \, \delta(u^N_i)
\end{align} 
for normalized weights 
\begin{equation*}
	w^N_i \equiv \frac{\Psi(u^N_i)}{ [\Psi(u^N_1) + \ldots + \Psi(u^N_N)]}
\end{equation*}
converges weakly to $\nu$. For concreteness, define the mapping from 
$\mu^N$ to $\nu^N$ as $\nu^N = \bayesop_{\Psi}(\mu^N)$ where $\bayesop_{\Psi}$ is the so-called 
Bayes operator that transforms a probability distribution $\mu$ into the probability 
distribution $\bayesop_{\Psi}(\mu)$ that satisfies 
$\bayesop_{\Psi}(\mu)(\phi) = \mu(\Psi \, \phi) / \mu(\Psi)$ for any test function $\phi$.
The following proposition shows that, under a mild {\it uniform integrability} condition, 
the convergence $\bayesop_{\Psi}(\mu^N) \wc \bayesop_{\Psi}(\mu)$ holds.
\begin{prop} \label{prop.consistency.reweighting}
Consider a probability distribution $\mu$ and a continuous and 
positive $\mu$-integrable function $\Psi$.
Assume that there exists a continuous $\mu$-integrable function 
$\mathcal{E}:\mathcal{X} \to [1,\infty)$ such that 
\begin{align} \label{eq.controlling.function}
\lim_{t \to \infty} \limsup_{N \to \infty} \; \mu^N\BK{ \mathcal{E} \times \mathbbm{1}_{\mathcal{E}>t}} = 0,
\end{align}
and $\Psi(u) \leq \mathcal{E}(u)$ for $\mu$-almost every $u \in \mathcal{X}$. We have that:
\begin{enumerate}
\item for any (potentially unbounded) continuous test function $\phi$ such that $|\phi| \leq \mathcal{E}$, 
$$\lim_{N \to \infty} \; \mu^N(\phi) = \mu(\phi).$$
\item the sequence $\bayesop_{\Psi}(\mu^N)$ converges weakly towards $\bayesop_{\Psi}(\mu)$.
\end{enumerate}
\end{prop}
\begin{remark}
The technical condition Equation \eqref{eq.controlling.function}
means that if $\zeta_N$ is a sequence of random variables such that $\zeta_N \sim \mu^N$, 
the sequence of scalar random variables $\overline{\zeta}_N \equiv \mathcal{E}(\zeta_N)$ is uniformly integrable \cite{Williams1991}.
\end{remark}
\begin{proof}
The second assertion is a direct consequence of the first one since 
\begin{gather*}
\bayesop_{\Psi}(\mu^N)(\phi) = \frac{\mu^N(\Psi \, \phi)}{\mu^N(\Psi)}
\qquad \textrm{and} \qquad
\bayesop_{\Psi}(\mu)(\phi)  = \frac{\mu(\Psi \, \phi)}{\mu(\Psi)}, 
\end{gather*}
and $\mu^N(\Psi) \to \mu(\Psi)$ as well as $\mu^N(\Psi \, \phi) \to \mu(\Psi \, \phi)$ for any 
bounded and continuous test function $\phi$. Let us now prove the first assertion.
Since $\mathcal{X}$ is a metric space and $\mathcal{E}$ is continuous, for any threshold $t \geq 0$ there exists (Urysohn's lemma) a separating continuous function $\rho_t: \mathcal{X} \to [0,1]$ (Urysohn's function) such that $\rho_t(u)=1$ on the set $\{u \in \mathcal{X} : \mathcal{E}(u) \leq t-1\}$ and $\rho_t(u)=0$ on the set $\{u \in \mathcal{X} : \mathcal{E}(u) \geq  t\}$. 
Since $\mathcal{E}$ is $\mu$-integrable and $|\phi| \leq \mathcal{E}$ $\mu$-almost everywhere, then for any $\epsilon > 0$ there exists $T_\epsilon \geq 0$ such that $|\mu(\phi) - \mu(\phi \, \rho_t)| < \epsilon$ for any $t \geq T_\epsilon$. Furthermore, since the function $\phi \, \rho_t$ is bounded and continuous and $\mu^N \wc \mu$, we have that $\mu^N(\phi \, \rho_t) \to \mu(\phi \, \rho_t)$. It follows that for any $t > T_\epsilon$ 
\begin{align*}
\limsup_{N \to \infty} |\mu^N(\phi) - \mu(\phi)| &\leq 
\limsup_{N \to \infty} |\mu^N(\phi \, \rho_t) - \mu(\phi)| + 
\limsup_{N \to \infty} |\mu^N(\phi \, (1-\rho_t))| \\
&\leq 
\limsup_{N \to \infty} |\mu^N(\phi \, \rho_t) - \mu(\phi)| + 
\limsup_{N \to \infty} \mu^N(\mathcal{E} \times \mathbbm{1}_{\mathcal{E} > t-1}) \\
&\leq
\epsilon + \limsup_{N \to \infty}  \mu^N(\mathcal{E} \times \mathbbm{1}_{\mathcal{E} > t-1}).
\end{align*}
Equation \eqref{eq.controlling.function} gives the conclusion.
\end{proof}

Note that if the potential $\Psi$ is bounded, Proposition \ref{prop.consistency.reweighting} always applies. 
In the standard Monte-Carlo setting where $u^N_i=u_i$ for i.i.d samples $\{u_i\}_{i \geq 0}$ from the distribution $\mu$, more precise estimates are available. The distributions $\mu^N$ and $\nu^N$ are random and one can readily check that 
\begin{equation}
  \label{muNconvergence}
  \vertiii{\mu^N - \mu} \leq \frac{1}{\sqrt{N}}, 
\end{equation}
where we have used the norm defined as
\begin{equation}
\label{eq.IS.bound}
	\vertiii{\mu^N - \mu}^2 \equiv \sup_{\norm{\phi}_{\infty} < 1} \, \E{ \BK{\mu^N(\phi) - \mu(\phi)}^2 }
\end{equation}
to measure the discrepancy between two random measures. 
Furthermore, \cite[Theorem $2.1$]{Agapiou2015} states that
\begin{align*}
	\vertiii{\mu^N - \mu} \leq \frac{2}{\sqrt{N}} \, \frac{ \mu\BK{\Psi^2 }^{\frac{1}{2}} }{ \mu(\Psi) }.
\end{align*}
The sequence of approximations $\mu^N$ converges at Monte-Carlo rate towards $\mu$.

\subsection{Resampling schemes} 
\label{subsec:multinomialresamp}
In standard SMC methods, as well as the SET method described in this article, one needs to transform a weighted particle approximation of a distribution $\mu$ into an equally 
weighted particle approximation of the same distribution. The multinomial resampling 
scheme  approximates
$\mu^N = \sum_{i=1}^N \, w^N_i \, \delta(u^N_i)$ by the equally weighted particle approximation
\begin{align*}
\mu^N_{\IS} \equiv \frac{1}{N} \sum_{i=1}^N \delta(u^N_{i,\IS})
\end{align*}
where $\{ u^N_{i,\IS} \}_{i=1}^N$ are i.i.d. samples from
$\mu^N$. Equation \eqref{muNconvergence} states that the norm between a distribution and an equally weighted mixture of Dirac masses centred at $N$ i.i.d samples from that distribution 
is less than $1 / \sqrt{N}$. Applying this remark and the fact that $\mu^N_{\IS}$ is precisely an 
equally weighted mixture of Dirac masses centred at $N$ i.i.d samples from $\mu^N$, it follows 
that $\vertiii{\mu^N_{\IS} - \mu^N} \leq 1 / \sqrt{N}$. There are 
more sophisticated approaches, such as the stratified
\cite{Hol2006} and systematic \cite{Douc2005} resampling methods, 
to generate equally weighted particle approximations.
We refer the reader to
\cite{gerber2017negative} for a recent study of theoretical properties
of these typically more statistically efficient resampling
schemes. Unless otherwise stated, all the numerical simulations
presented in this article use the stratified resampling scheme.

For concreteness, we denote by $\resampop$ the resampling operator that
maps a weighted particle approximation to an equally weighted
one. Note that for a given weighted particle approximation
$\mu^N$, the quantity $\resampop(\mu^N)$ is in general a random
probability distribution. The resampling scheme $\resampop$  is called {\it consistent} 
if it maps $\mu^N$, 
a possibly random sequence of distributions that almost surely converges weakly 
towards $\mu$, into another sequence $\resampop(\mu^N)$ 
that almost surely converges weakly towards $\mu$.
It has long been known
\cite{Crisan2002} that the multinomial resampling scheme is consistent in finite 
dimensional Euclidean spaces. As investigated in \cite{Hol2006}, the situation is much more delicate for the stratified and systematic resampling methods.  

\subsection{Mutation} 
\label{subsec:mutation}
Consider a sequence $\{\mu_k\}_{k=0}^K$ of distributions interpolating between a tractable distribution $\mu_0$ and the posterior distribution $\mu_K$ such that for any index $1 \leq k \leq K$ we have
\begin{align*}
\frac{d \mu_k}{d \mu_{k-1}}(u) = \frac{1}{Z_k} \, \Psi_k(u)
\end{align*}
for a $\mu_{k-1}$-integrable potential function $\Psi_k: \cX \to (0,\infty)$. 
For technical reasons, we also assume that $\Psi_k$ is continuous. 
Consider a particle approximation
\begin{align*}
\mu^N_0 = \frac{1}{N} \,  \sum_{i=1}^N \delta(u^N_{0,i})
\end{align*}
of the initial distribution $\mu_0$.
Under mild assumptions, the sequence of equally weighted 
distributions $\mu^N_k = (1/N) \, \sum_{i=1}^N \delta(u^N_{k,i}) $ recursively defined as $\mu^N_k = \resampop \circ \bayesop_{\Psi_k}(\mu^N_{k-1})$ converges in an appropriate sense towards $\mu_k$ as $N \to \infty$. For example, Proposition \ref{prop.consistency.reweighting} shows that, if the potential $\Psi_k$ are bounded and the resampling scheme $\resampop$ is consistent, as soon as $\mu^N_0$ almost surely converges weakly towards $\mu_0$ the sequence $\mu^N_k$ also almost surely converges weakly towards $\mu_k$ as $N \to \infty$.

In most realistic scenarios, though, 
the particle approximation $\mu^N_K$, as an approximation to $\mu_K$, is worse  than 
the direct importance sampling particle approximation $\bayesop_{\Psi_1 \Psi_2 \ldots \Psi_K}(\mu^N_0)$ 
from $\mu_0$ to $\mu_K$ where $(d \mu_K / d \mu_0)(u) \propto [\Psi_1 \Psi_2 \ldots \Psi_K](u)$. 
It is because in that case 
the particles $\{u^N_{K,i}\}_{i=1}^N$ form a subset of $\{u^N_{0,i}\}_{i=1}^N$. 
Consequently, if the initial set of particles $\{u_{0,i}\}_{i=1}^N$ are located in regions of the parameter space where the distribution $\mu_K$ does not have much probability mass, the approximation $\mu^N_K$ to $\mu_K$ can be very poor. 
For importance sampling to work well in high-dimensional situations, the proposal distributions 
need to be chosen very judiciously, and  adaptive importance sampling (AIS) 
\cite{oh1992adaptive, cornuet2012adaptive, cappe2008adaptive, finke2019relationship} can partially remedy this issue.
A standard approach to mitigate this issue 
is to introduce mutation steps, which we now describe. 
For each distribution $\mu_k$ in the interpolating sequence 
of distributions, consider a (mutation) 
Feller Markov kernel $M_k(u, d\widehat{u})$ that leaves the distribution $\mu_k$ invariant. 
Consider the operator $\markovop_k$ 
that transforms a particle approximation $\mu_k^N = (1/N) \, \sum_{i=1}^N \delta(u^N_{k,i})$ into 
$\markovop_k(\mu_k^N) = (1/N) \, \sum_{i=1}^N \delta(v^N_{k,i})$ where, conditionally 
upon $\{u^N_{k,i}\}_{i=1}^N$, the samples 
$\{v^N_{k,i}\}_{i=1}^N$ are independent realizations of $M_k(u^N_{k,i}, d\widehat{u})$. 
The following lemma shows that, as soon as the sequence $\mu_k^N$ almost surely converges  weakly to $\mu_k$, the sequence $\markovop_k(\mu_k^N)$ also almost surely converges weakly to $\mu_k$.\\

\begin{lemma} \label{lem.ergodic.markov}
Let $\mu$ be a probability distribution on a locally compact and $\sigma$-compact metric space $\cX$. Consider $M(u, d\widehat{u})$ a $\mu$-invariant Feller Markov kernel.
For each $N \geq 1$, let $\{u^N_i\}_{i=1}^N \subset \cX$ be such that
\begin{align*}
\frac{1}{N} \sum_{i=1}^N \delta(u^N_i)  \wc \mu.
\end{align*}
For independent random variables $V^N_i \sim M(u^N_i, d\widehat{u})$, we have that, almost surely,
\begin{equation*}
\frac{1}{N} \sum_{i=1}^N \delta(V^N_i) \wc \mu. 
\end{equation*}
\end{lemma}
\begin{proof}
Since $\cX$ is a locally compact and $\sigma$-compact metric space, there exists a countable and dense (for the supremum norm) subset  $\mathcal{H}$ of the set of continuous functions with compact support in $\cX$. One needs to prove that for any $\phi \in \mathcal{H}$ we have that $\lim\limits_{N \to \infty} \, (1/N) \sum_{i=1}^N \phi(v^N_i) = \mu(\phi)$ almost surely. Since the 
function $M \phi$ is continuous and bounded,
\begin{align*}
\lim_{N \to \infty}  \, \mathbb{E}\sqBK{ \frac{1}{N} \sum_{i=1}^N \phi(v^N_i) }
&=
\lim_{N \to \infty}  \, \frac{1}{N} \sum_{i=1}^N (M \phi)(u^N_i) =
\mu( M \phi ) = \mu(\phi).
\end{align*}
Since $\phi$ is bounded, the moment of order four of the ergodic sum $\frac{1}{N} \sum_{i=1}^N [\phi(v^N_i) - (M \phi)(u^N_i)]$ is upper bounded by a constant multiple of $N^{-2}$. The Borel-Cantelli lemma gives the conclusion.
\end{proof}

Leveraging these Markov mutation kernels, we now define the sequence of equally 
weighted particle approximations 
$\{ \mu^N_k \}_{k=0}^K$ recursively as
\begin{equation}
\label{SMC}
\mu^N_k = \markovop_k \circ \resampop \circ \bayesop_{\Psi_k} (\mu^N_{k-1}).
\end{equation}
The Markov mutations ensure that, in general, the particles $\{u^N_{k,i}\}_{i=1}^N$ do not form a subset of $\{u^N_{0,i}\}_{i=1}^N$.
The particle algorithm resulting from \eqref{SMC} is a special case of
Sequential Monte Carlo (SMC) samplers \cite{DelDoucetJasra06}.  Note that, 
in Bayesian inverse problems, simulating from the Markovian 
kernel $M_k$ typically requires evaluating the computationally expensive 
forward map. Moreover, as explained in the introduction, whilst well-designed 
Markovian kernels can greatly enhance the statistical efficiency of the resulting 
algorithm, it is notoriously difficult to design well-mixing mutation 
kernels in high-dimensional settings or fir exploring distribution with complex dependency structures.

\section{Optimal Transport}
\label{subsec:optimaltrans}
For technical simplicity, we assume in this section that the state
space $\cX$ is a finite dimensional Euclidean space with norm denoted
by $\| \, \cdot \, \|$.  For two distributions $\mu$ and $\nu$ related by a change 
of probability $d\nu / d\mu (u) \propto \Psi(u)$, the Monge-Kantorovich optimal transport
approach provides an alternate methodology for building a particle approximation 
of a distribution $\nu$ out of a particle approximation of $\mu$. 
To the best of our knowledge, the idea
was first proposed in \cite{Reich2013}, and further developed in
\cite{Gregory2016,Chustagulprom2016,graham2019scalable}, in the context of
data-assimilation of dynamical systems.
For two probability distributions $\mu$ and $\nu$, let 
$\mathcal{P}(\mu,\nu)$ be the set of probability couplings between $\mu$ and $\nu$,
i.e. the convex set of probability distributions on $\cX \times \cX$
that admit $\mu$ and $\nu$ as marginals. For a cost function
$\cost:\cX \times \cX \to [0,\infty)$, the optimal transportation 
problem seeks to minimize the transport cost 
$\mathbb{E}_{\gamma}\sqBK{\cost(\hat{u},\hat{v})}$, for $(\hat{u},\hat{v}) \sim \gamma$, 
over the set of all possible couplings $\gamma \in \mathcal{P}(\mu, \nu)$,
\begin{align}
	\gamma^{\OT} \; = \; \argmin \Big\{ \gamma \mapsto \mathbb{E}_{\gamma}\sqBK{ \cost(\hat{u},\hat{v})}
\quad \textrm{with} \quad
\gamma \in \mathcal{P}(\mu,\nu) \Big\}. 
\label{eqn:originalDOT}
\end{align}
On an Euclidean space, a standard choice is the quadratic cost function $\cost(u,v) = \|u-v\|^2$. 
For cost functions of the type $\cost(u,v) = h(v-u)$ for a strictly convex function $h$, Brenier's theorem \cite{Brenier1991} states that, 
if $\mu$  is compactly supported and has a density with respect to the Lebesgue measure, 
there exists a  deterministic map $\mathbf{T}:\cX \to \cX$, uniquely defined on the support of $\mu$, such that the optimal coupling $\gamma^{\OT}$ is obtained by pushing-forward the distribution $\mu$ through the deterministic function $(\mathbf{Id}, \mathbf{T}): \cX \mapsto \cX \times \cX$. 
That is, for a test function $\phi: \cX \times \cX \to \mathbb{R}$, the quantity 
$\gamma^{\OT}(\phi)$ can also be expressed as 
$\mathbb{E}_{\mu}[\phi(\hat{u}, \mathbf{T}(\hat{u}) )]$ for $\hat{u} \sim \mu$. For more general cost functions, 
the situation is more delicate \cite{evans1999differential,trudinger2001monge,caffarelli2002constructing,Ambrosio2003a}.

\subsection{Approximation of the Bayes operator}
Consider a weighted particle approximation $\mu^N = \sum_{i=1}^N \alpha_i \,
\delta(u^N_i)$ of the distribution $\mu$ and, for a potential function $\Psi:\mathcal{X} \to (0; \infty)$, 
the probability distribution
\begin{equation}
\bayesop_\Psi(\mu^N) \equiv \sum_{i=1}^N \beta_i \, \delta(u^N_i) \equiv \nu^N,
\label{betaApprox}
\end{equation}
with $\beta_i = \alpha_i \, \Psi(u^N_i) / [\alpha_1 \, \Psi(u^N_1) + \ldots + \alpha_N \, \Psi(u^N_N)]$.
The optimal coupling $\gamma^{\OT,N}$ between $\mu^N$ and $\nu^N$ is supported on the finite set 
$\{(u^N_i, u^N_j)\}_{1 \leq i, j \leq N}$ and can thus be expressed as
\begin{equation*}
\gamma^{\OT,N} = \sum_{i,j=1}^N \C^{\OT,N}_{ij} \, \delta(u^N_i) \otimes \delta(u^N_j).
\end{equation*}
where $\delta(u) \otimes \delta(v)$ denotes the Dirac mass centred at $(u,v) \in \cX \times \cX$.
Here, the coupling matrix $\C^{\OT,N} \in \mathbb{R}_+^{N,N}$ is the solution of the linear programming problem that consists in minimizing the matrix functional 
\begin{align} \label{eq.ot.map}
C \mapsto \sum_{i,j=1}^N C_{i,j} \times \cost(u^N_i, u^N_j) \equiv \bra{ C, \costmat}_{\text{F}}
\qquad \text{with} \qquad
\costmat_{i,j} = \cost(u^N_i, u^N_j)
\end{align}
over the convex set $\mathcal{P}(\alpha, \beta)$ of matrices with marginals $\alpha$ and $\beta$, i.e. the set of matrices $C \in \mathbb{R}_+^{N,N}$ such that $\sum_j C_{i_0,j} = \alpha_{i_0}$ and $\sum_i C_{i,j_0} = \beta_{j_0}$ for all $1 \leq i_0, j_0 \leq N$. In Equation \eqref{eq.ot.map}, the quantity $\bra{ C, \costmat}_{\text{F}} = \sum_{i,j} C_{i,j} \, \costmat_{i,j}$ is the Frobenius inner product between the coupling matrix $C$ and the {\it cost matrix} $\costmat \in \mathbb{R}^{N,N}$.
More details are given at the end of this section.

We now describe how, once the coupling matrix $\C^{\OT,N}$ has been computed, a particle approximation of the distribution $\bayesop_\Psi(\mu^N)$ can be constructed: we stress that, in order to implement this method, the coupling matrix $\C^{\OT,N}$ is the only quantity that needs to be computed.
For motivating the methodology, assume that the optimal coupling $\gamma^\OT \in \mathcal{P}(\mu, \nu)$ 
is described by a deterministic map
$\mathbf{T}:\cX \to \cX$ and consider a test function $\phi: \cX \to
\mathbb{R}$. 
Since $\mu^N$ is a particle approximation to $\mu$, 
the quantity $\mu^N(\phi \circ \mathbf{T})=\sum_{i=1}^N \alpha_i \, \delta(\mathbf{T}(u^N_i))$ 
is expected to be an approximation of $\mu(\phi \circ \mathbf{T}) = \nu(\phi)$.  
Consequently, it is reasonable to expect
\begin{align*}
\sum_{i=1}^N \alpha_i \, \delta(\mathbf{T}(u^N_i))
\end{align*}
to be a particle approximation of $\nu$. Although the optimal transformation 
$\mathbf{T}$ is generally computationally intractable (i.e. it is never actually computed in our proposed method) one can resort to an approximation scheme. Note that the quantity
$\mathbf{T}(u^N_i)$ can be expressed as a conditional expectation
\begin{align*}
	\mathbf{T}(u^N_i) 
	\; = \; \mathbb{E}[\hat{v} \mid \hat{u} = u^N_i]
	\qquad \textrm{for} \qquad (\hat{u},\hat{v}) \sim \gamma^{\OT},
\end{align*}
since the pair $(\hat{u},\hat{v})$ has the same distribution as $(\hat{u}, \mathbf{T}(\hat{u}))$ for $\hat{u} \sim \mu$. 
This motivates the approximation
\begin{align} \label{eq.particle.transport}
\mathbf{T}(u^N_i) \; \approx \;  \mathbb{E}\sqBK{ \hat{v}^N \mid \hat{u}^N = u^N_i } 
= 
\frac{\sum_{j=1}^N \C^{\OT,N}_{ij} \, u^N_j}{\sum_{j=1}^N \C^{\OT,N}_{ij} } 
= 
\frac{1}{\alpha_i } \, \sum_{j=1}^N \C^{\OT,N}_{ij} \, u^N_j
\end{align}
with $(\hat{u}^N, \hat{v}^N) \sim \gamma^{\OT,N}$. The newly created particles 
$\{ u^{\OT,N}_i \}_{i=1}^N$ defined as
\begin{align} \label{eq.ot.particle.map}
u^{\OT,N}_i \equiv \frac{1}{\alpha_i } \, \sum_{j=1}^N \C^{\OT,N}_{ij} \, u^N_j
\end{align}
are convex combinations of the original particles $\{u^N_1, \ldots, u^N_N\}$ and thus all lie in the convex hull of the set of original particles. 
In summary, the computational optimal transport executed in the SET algorithm proceeds by first solving for $\C^{\OT,N}$ given the constraints described by Equation \ref{eq.ot.map}. In a second stage, the coupling matrix $\C^{\OT,N}$ is then used to transport the particles following Equation \ref{eq.ot.particle.map}. 
For concreteness and in accordance with the previous sections, we denote by $\OTop_\Psi$ the operator that realizes the mapping
\begin{equation}
\label{eq:otOp}
\OTop_\Psi \left ( \sum_{i=1}^N \alpha_i \, \delta(u^N_i) \right ) \equiv 
\sum_{i=1}^N \alpha_i \, \delta(u^{\OT,N}_i) =
\sum_{i=1}^N \alpha_i \, \delta\BK{ \frac{1}{\alpha_i } \, \sum_{j=1}^N \C^{\OT,N}_{ij} \, u^N_j }.
\end{equation}
Similar to the operator $\resampop \circ \bayesop_{\Psi}$, the operator $\OTop_\Psi$ maps an equally weighted particle approximation of a probability distribution $\mu$ into an equally weighted particle approximation of $\bayesop_\Psi(\mu)$. However, unlike $\resampop \circ \bayesop_{\Psi}$, the support of the particle approximation $\mu^N$ and $\OTop_\Psi(\mu^N)$ are typically disjoint.\\
%
%
%
\begin{algorithm}[H]
\label{algo.ot.transport}
\caption{Optimal Transportation operator $ \OTop_\Psi$}
{\bf Weights computation:} define 
$$\beta_i = \alpha_i \, \Psi(u^N_i) / [\alpha_1 \, \Psi(u^N_1) + \ldots + \alpha_N \, \Psi(u^N_N)].$$
{\bf Cost matrix:} build the matrix $\costmat \in \mathbb{R}^{N,N}$ defined in \eqref{eq.ot.map}.\\
{\bf Optimal Transport:} compute $\C^{\OT,N} = \argmin_{\mathcal{P}(\alpha, \beta)} C \mapsto \bra{C, \costmat}_{\text{F}}$.\\
{\bf Transportation:} set $u_i^{\OT,N} = (1 / \alpha_i ) \, \sum_{j=1}^N \C^{\OT,N}_{ij} \, u^N_j$ and define
\begin{align*}
 \OTop_\Psi(\mu^N) = 
\sum_{i=1}^N \alpha_i \, \delta(u^{\OT,N}_i) 
\end{align*}
\end{algorithm}
Algorithm \ref{algo.ot.transport} summarizes the optimal transport approach to approximating the Bayes operator that transforms a particle approximation $\mu^N = \sum_{i=1}^N \alpha_i \,
\delta(u^N_i)$ of a distribution $\mu$ into a particle approximation $\OTop_\Psi(\mu^N)$ of the distribution $\nu = \bayesop_\Psi(\mu)$,
\begin{align*}
\mu^N = \sum_{i=1}^N \alpha_i \,\delta(u^N_i)
\quad 
\autorightarrow{\vbox{\hbox{\scriptsize Optimal Transport}\vskip-5pt}}{\vbox{\hbox{$\OTop_\Psi$}\vskip-5pt}} 
\quad
\sum_{i=1}^N \alpha_i \, \delta(u^{\OT,N}_i) \equiv \OTop_\Psi(\mu^N).
\end{align*}
The only potentially computationally expensive step is the computation of the coupling matrix $\C^{\OT,N}$. The computational costs are discussed at the start of Section \ref{sec:numerical} and we refer the reader to \cite{peyre2019computational} for a book-length treatment of the computational aspects associated to optimal transportation problems.

\subsection{Consistency} \label{sec.consistency}
Consider a potential function $\Psi:\mathcal{X} \to (0,\infty)$ and
two distributions $\mu$ and $\nu = \bayesop_\Psi(\mu)$. In this
section, we generalize and extend Theorem $1$ of \cite{Reich2013} to
prove that, under mild assumptions, the optimal transport operator
$\OTop_\Psi$ transforms a sequence $\mu^N \wc \mu$ into a sequence
$\OTop_\Psi(\mu^N)$ that converges weakly to $\bayesop_\Psi(\mu)$. 

%
%
\begin{assumption} [Unique Deterministic Coupling]
\label{assump.deterministic.coupling} 
The optimal transport problem between $\mu$ and $\bayesop_\Psi(\mu)$ with cost function $\cost$ admits a unique solution $\gamma$ that can be realized by a deterministic transport map $\mathbf{T}:\cX \to \cX$. \\
\end{assumption}

The problem of existence and uniqueness of the solution to an optimal
transport problem is well-studied. Under mild assumptions (see
McCann's main theorem \cite{Mccann1995}), the set of couplings between
$\mu$ and $\nu$ is weakly compact and the functional $\mu \mapsto
\mathbf{E}_\mu[\cost(u,v)]$ is continuous in the appropriate topologies,
ensuring the existence of an optimal coupling.  The uniqueness and
regularity properties of the optimal transport map are more delicate to establish 
and we refer to \cite{Cavalletti2015} for recent developments. To proceed to
the main result of this section we further assume the following.

%
%
\begin{assumption}[Regularity of the Transport Map] \label{assump.reg}
Let Assumption \ref{assump.deterministic.coupling} holds for a 
deterministic map $\mathbf{T}:\cX \to \cX$. For any bounded and Lipschitz 
function $\phi:\cX \to \mathbb{R}$ and sequence $\mu^N$ that converges 
weakly to $\mu$, 
we have that $\mu^N(\phi \circ \mathbf{T}) \to \mu(\phi \circ \mathbf{T})$.\\
\end{assumption}

The continuous mapping theorem \cite{Mann1943}
shows that Assumption \ref{assump.reg} is satisfied provided that
the set of discontinuities of $\mathbf{T}$ has zero measure under $\mu$. In particular, Assumption \ref{assump.reg}  holds in the case when the optimal map $\mathbf{T}$ is continuous.
Theorem \ref{thm.consistency} below shows that, under mild growth and regularity assumptions on the optimal transport map $\mathbf{T}:\cX \to \cX$, the optimal transport scheme $\OTop_\Psi$ is consistent as the number of particles $N \geq 1$ approaches infinity.

\begin{theorem} \label{thm.consistency}
Consider a potential function $\Psi:\cX \to (0;\infty)$ and two probability distributions $\mu$ and $\nu = \bayesop_\Psi(\mu)$ on the state space $\cX$. Assume that Assumptions \ref{assump.deterministic.coupling} and \ref{assump.reg} are satisfied for a deterministic optimal map $\mathbf{T}:\cX \to \cX$. Consider further a sequence of weighted particle approximations
\begin{align*}
	\mu^N = \sum_{i=1}^N \alpha^N_i \, \delta(u^N_i)
\end{align*}
that converges weakly to $\mu$, and such that $\bayesop_\Psi(\mu^N)$ 
converges weakly to $\bayesop_\Psi(\mu)$. If the growth assumption
\begin{align}
\label{eq.uniform.integrability}
\limsup_{N \to \infty} \quad \mu^N(u \mapsto |\mathbf{T}(u)|^p) + \bayesop_\Psi(\mu^N)(u \mapsto |u|^p) < \infty,
\end{align}
is satisfied for some exponent $p>1$, we have that 
\begin{align}
\label{eq.otconv}
	\OTop_\Psi(\mu^N) \wc \bayesop_\Psi(\mu) \equiv \nu.
\end{align}
\end{theorem}
\begin{proof}
Let $\gamma^{\OT,N} = \sum_{i,j} \C^N_{i,j} \, \delta(u^N_i) \otimes \delta(u^N_j)$ be the optimal coupling between $\mu^N$ and  $\bayesop_\Psi(\mu^N)$.
By assumption, $\mu^N \wc \mu$ and $\nu^N \equiv \bayesop_\Psi(\mu^N) \wc \bayesop_\Psi(\mu) \equiv \nu$ and there is a unique optimal  coupling $\gamma^{\OT}$ between $\mu$ and $\nu$. By compactness (see, e.g. \cite[Corollary $5.21$]{Villani2008}), we have that $\gamma^{\OT,N} \wc \gamma$ as $N \to \infty$.

To show the  weak convergence of $\OTop_\Psi(\mu^N)$ towards $\nu$, it suffices to prove that for any Lipschitz and bounded test function $\phi$ we have that 
$\OTop_\Psi(\mu^N) \to \nu(\phi)$. Assumption \ref{assump.reg} implies 
$\mu^N(\phi \circ \mathbf{T}) \to \mu(\phi \circ \mathbf{T}) = \nu(\phi)$. Consequently, it suffices to show that the difference 
$\OTop_\Psi(\mu^N) - \mu^N(\phi \circ \mathbf{T})$ converges to zero as $N \to \infty$, i.e.,
%
\begin{align*}
\lim_{N \to \infty} \, \sum_{i=1}^N \alpha^N_i \, \left| \phi\BK{ \frac{1}{\alpha^N_i} \, \sum_{j=1}^N \C^N_{ij} \, u^N_j} - \phi\BK{ \mathbf{T}(u^N_i) } \right| = 0.
\end{align*}
Since $\phi$ is Lipschitz, and $ \sum_{j=1}^N \C^N_{ij} = \alpha^N_i$,  it is sufficient to show that 
\begin{align*}
\lim_{N \to \infty} \, \sum_{i,j}^N \C^N_{ij} \, \left| u^N_j -  \mathbf{T}(u^N_i) \right| 
= 0.
\end{align*}
Note that $\sum_{i,j}^N \C^N_{ij} \, \left| u^N_j -  \mathbf{T}(u^N_i) \right| = \gamma^{\OT,N}(F)$ with $F(u,v) = |v - \mathbf{T}(u)|$. Since $F^p(u,v) \lesssim |v|^p + |\mathbf{T}(u)|^p$, 
assumption \eqref{eq.uniform.integrability} yields that $\limsup_N \, \gamma^{\OT,N}(F^{p}) < \infty$.
Since $\gamma^{\OT,N} \wc \gamma$, the bound $\limsup_N \, \gamma^{\OT,N}(F^{p}) < \infty$ implies that the sequence $\gamma^{\OT,N}(F)$ converges towards $\gamma(F)$.
Since $\gamma(F)=0$, the conclusion follows.
\end{proof}

%
%
\section{Sequential Ensemble Transform}
\label{subsec.SET}
In this section, we describe our proposed methodology, the {\it Sequential Ensemble Transform} (SET), 
prove that it is consistent in the limit of infinitely many particles, and discuss adaptation strategies that 
are important for practical implementations of the method.

\subsection{High-level description and consistency}
As in Section \ref{sec:particlemethods}, consider a sequence
$\{\mu_i\}_{i=0}^K$ of distributions that interpolates between a
distribution $\mu_0$  and the posterior distribution $\mu_K$. For any index $1 \leq k \leq K$ we have that $(d \mu_k
/ d \mu_{k-1})(u) = (1 / Z_k) \, \Psi_k(u)$ for a
$\mu_{k-1}$-integrable and continuous potential function $\Psi_k: \cX
\to (0,\infty)$. In this section, we assume the following.
%
%
\begin{assumption} 
\label{assump.SOT} 
The sequence of probability distributions $\{\mu_k\}_{k=0}^K$ is such that:
\begin{enumerate}
\item for any $0 \leq k \leq K$, the support of $\mu_k$ is bounded,
\item for any $1 \leq k \leq K$, the pair of distributions $(\mu_{k-1}, \mu_k)$ satisfies Assumptions \ref{assump.deterministic.coupling} and \ref{assump.reg}.
\end{enumerate}
\end{assumption}
Instead of constructing a sequence of particle approximations to the
intermediate distributions $\mu_k$ through importance
sampling-resampling methods, consider the following approach that
leverages optimal transport. Let $\mu_0^N = (1/N) \, \sum_{i=0}^N
\delta(u^N_{0,i})$ be an equally-weighted particle approximation of
the initial distribution $\mu_0$. Define the equally weighted particle
approximations $\mu^N_k$ through the recursion formula
\begin{align} \label{eq.SOT}
\mu^N_k = \markovop_k \circ \OTop_{\Psi_k} (\mu^N_{k-1}),
\end{align}
where $\markovop_k$ is the operator associated to a $\mu_k$-invariant Markov mutation kernel $M_k$.
%
%
\begin{theorem}[Consistency of the SET algorithm] \label{thm.SOT.consistency}
  Let $\{\mu_k\}_{k=0}^K$ be a sequence of distributions that satisfies Assumption \ref{assump.SOT} and consider $\{u^N_{0,i}\}_{i=1}^N \subset \mathbb{R}^d$ such that
\begin{align*}
\mu_0^N \equiv \frac{1}{N} \sum_{i=1}^N \delta(u^N_{0,i}) \; \wc \; \mu_0.
\end{align*}
Then, for any index $1 \leq k \leq K$, the sequence of equally weighted particle approximations $\mu^N_k$ defined recursively through Equation \eqref{eq.SOT} weakly converges to
$\mu_k$ almost surely.
\end{theorem}

\begin{proof}
One can proceed by induction. It suffices to prove that if
$\mu^N_{k-1} \wc \mu_{k-1}$ almost surely then
$\markovop_k \circ \OTop_{\Psi_k} (\mu^N_{k-1}) \equiv \mu^N_k \wc
\mu_k$ almost surely. Under Assumption \eqref{assump.SOT}, the support of the
distribution $\mu_{k-1}$ is bounded: one can find a bounded and
continuous function $V_k$ that dominates $\Psi_k$ and invoke
Proposition \ref{prop.consistency.reweighting} to see that
$\bayesop_{\Psi_k}(\mu^N_{k-1}) \wc \mu_{k}$ almost
surely. Furthermore, under Assumption \ref{assump.SOT} the pair
$(\mu_{k-1}, \mu_k)$ satisfies Assumptions
\eqref{assump.deterministic.coupling}-\eqref{assump.reg} as well as
Equation \ref{eq.uniform.integrability}. Theorem \ref{thm.consistency} 
shows that $\OTop_{\Psi_k} (\mu^N_{k-1}) \wc \mu_{k}$ almost
surely. Finally, since the Feller Markov process $M_k$ lets $\mu_k$
invariant, Lemma \ref{lem.ergodic.markov} yields that $\markovop_k
\circ \OTop_{\Psi_k} (\mu^N_{k-1}) \wc \mu_k$ almost surely.\\
\end{proof}

As previously mentioned, one of the advantages of relying on optimal transportation instead of sampling-resampling techniques is that, as illustrated in Section \ref{sec:numerical}, the resulting algorithm is much less sensitive to the mixing properties of the Markov mutation kernels $M_k$. Moreover, the adaptive tempering strategies described in Section \ref{sec:tempmut} can be used within the SET method. In Section \ref{sec:numerical}, we compare the SET approach to more standard SMC approaches.

\subsection{Adaptive tempering}
\label{sec:tempmut}
In complex scenarios such as Bayesian inverse problems, it is a nontrivial task to specify a sequence of distributions \eqref{eq.bridge} 
that interpolates between a distribution $\mu_0$ that is straightforward to sample from and 
the posterior distribution. 
Instead, we consider an adaptive annealing scheme \cite{duetscher2000articulated,minvielle2010bayesian,jasra2011inference,zhou2016toward,nguyen2016efficient,schafer2013sequential,Kantas2014}. The reader is referred to \cite{Beskos2015a,giraud2017nonasymptotic} for a theoretical analysis of adaptive annealing methods.
For notational convenience, we identify distributions with their densities,
and assume that the posterior distribution $\mu_{\post}$ is absolutely continuous 
with respect to $\mu_0$, i.e. $d\mu_{\post} / d \mu_0(u) \propto \exp[V(u)]$ for some potential function $V: \mathbb{R}^d \to \mathbb{R}$. Consider the sequence $\{\mu_k\}_{k=0}^K$ defined as
\begin{align} \label{eq.tempered.dist}
\frac{d \mu_k}{d \mu_0}(u) \; \propto \;
\exp\sqBK{ \tau_k \, V(u) }
\end{align}
for an (inverse) temperature parameter $\tau_k$ that interpolates between $\tau_0=0$ and 
$\tau_K = 1$. 
In practice, it can be difficult to choose the number $K \geq 1$ of temperatures (i.e. the number of interpolating densities) and the corresponding temperatures. The adaptive scheme proceeds as follows. Assume that the particle approximation
\begin{align*}
\mu^N_{k} = \frac{1}{N} \, \sum_{i=1}^N \delta(u^N_{k,i})
\end{align*}
to the density $\mu_{k}$ has already been constructed. For a predetermined threshold $0 < \xi_{\ess} < 1$, the next temperature $\tau_{k+1}$ is defined as the smallest temperature $\tau > \tau_{k}$ such that $\ess_k(\tau) \leq \xi_{\ess}$ . Here, The {\it Effective Sample Size} (ESS) functional is defined as
%
%
\begin{align} \label{ess}
\ess_{k}(\tau) 
&\equiv
\frac{1}{N} \, \frac{\BK{\sum_{i=1}^N \exp\sqBK{ (\tau-\tau_k) \, V(u^N_{k,i})} }^2}
{\sum_{i=1}^N \exp\sqBK{ (\tau-\tau_k) \, V(u^N_{k,i})}^2} \in [0,1].
\end{align}
Clearly, $\ess_k(\tau_k) = 1$.
Lemma $3.1$ of \cite{Beskos2015a} states that the function $\tau \mapsto \ess_k(\tau)$ is decreasing for $\tau \in (\tau_k, \infty)$ 
so that $\tau_{k+1}$ can very efficiently be found by a bisection method. Finding $\tau_{k+1}$ 
typically does not require evaluating the forward map since, in standard implementations 
of the SMC or SET methods, the quantities $V(u^N_{k,i})$ would have already been computed 
at previous steps. Starting from $\tau_0 = 0$ and setting
\begin{align} \label{eq.tau.search}
\tau_{k+1} = \inf \curBK{\tau > \tau_k \; : \; \ess_k(\tau) \leq \xi_{\ess}},
\end{align}
the procedure stops as soon as $\tau_k$ is greater or equal to one. One thus sets 
$K = \inf \curBK{k \geq 1 \; : \; \tau_k \geq 1}$ and defines $\tau_K = 1$.
Note that taking $\xi_{\ess}$ close to one leads to a slow annealing, which may be computationally wasteful. 
On the other hand, taking $\xi_{\ess}$ close to zero can lead to an annealing scheme that is too rapid, ultimately 
leading to a poor particle approximation of the posterior distribution.
Except stated otherwise, we choose $\xi_{\ess} = 1/2$ in the numerical experiments of Section \ref{sec:numerical}.

%
%
\subsection{Adaptive mutation kernels}
\label{sec:markovmut}

Choosing a-priori a sequence of well-mixing Markov mutation kernels is, in most
realistic scenarios, not feasible. 
A standard approach consists in exploiting the population $\{u^N_{k,i}\}_{i=1}^N$ of particles at temperature $\tau_k$ to estimate summary statistics of
the distribution $\mu_k$. 
These summary statistics estimates (e.g. mean and covariance matrix) can then
be leveraged to design a Markov kernel $M_k$ with reasonable mixing
properties and that lets the distribution $\mu_k$ invariant.
In high-dimensional settings, this adaptive tuning of the mutation kernel is often crucial to obtaining satisfying performances. 
In this section, we concentrate on two classes of proposals, namely {\it autoregressive proposals} 
that do not make use of any derivative information and {\it Preconditioned Crank-Nicholson Langevin proposals} 
that can make use of gradient information for enhanced mixing properties. We refer the reader to \cite{CuiLawMarzouk16} and the references therein for more advanced adaptation strategies especially designed to tackle high-dimensional  Bayesian inverse problems.\\

\noindent
{\bf Autoregressive Proposals:} for a mean vector $\mathbf{m} \in \mathbb{R}^d$ and a positive definite covariance matrix $\mathbf{\Gamma} \in \mathbb{R}^{d,d}$, the Markovian proposal $u \mapsto \widehat{u}$ defined as
\begin{align} \label{eq.pcn}
\widehat{u} = \mathbf{m} + \rho \, \BK{u - \mathbf{m}} + (1-\rho^2)^{1/2} \, \mathcal{N}\BK{ 0, \mathbf{\Gamma} }
\end{align}
for some scaling factor $\rho \in (0,1)$ is reversible with respect to the Gaussian distribution with mean $\mathbf{m}$ and covariance $\mathbf{\Gamma}$. This proposal mechanism, also sometimes called the {\it Preconditioned Crank-Nicholson} proposal \cite{Cotter2013}, can consequently be used within a standard Metropolis-Hastings scheme to efficiently explore distributions that are well approximated by a Gaussian distribution with mean $\mathbf{m}$ and covariance $\mathbf{\Gamma}$. This remark can be used to design an adaptation strategy \cite{Kantas2014} for automatically tuning the mutation kernels within SMC methods or the SET algorithm. At iteration $k$, right after the resampling step of a SMC method, or right after the transportation step of the SET, consider a set $\{\widetilde{u}^N_{k,i}\}_{i=1}^N$ of particles whose (equally weighted) empirical distribution approximates the distribution $\mu_k$. In order to use an autoregressive Markov kernel \eqref{eq.pcn}, one can use the particles $\{\widetilde{u}^N_{k,i}\}_{i=1}^N$ to compute an approximation $\mathbf{m}^N_k$ of the mean of $\mu_k$ as well as an approximation $\mathbf{\Gamma}^N_k$ of its covariance matrix. In high-dimensional settings, or when the number of particles is low when compared to the dimensionality of the state-space, it is customary to only consider diagonal approximations of the covariance structure: the approximate covariance matrix $\mathbf{\Gamma}^N_k$ is diagonal, with the empirical marginal variances on its diagonal. The scaling factor $\rho^N_k$ can also be chosen adaptively. Values of $\rho^N_k \approx 1^-$ lead to conservative proposals while values of $\rho^N_k \approx 0^+$ are more likely to be rejected. Given two fixed thresholds $0<\xi_- < \xi_+ < 1$, the scaling factor $\rho^N_k$ can be adapted based upon the acceptance rate of the Metropolis-Hastings proposals \eqref{eq.pcn}. Specifically, we set 
$\rho^N_{k} = \min(1, [1+\epsilon] \, \rho^N_{k-1})$ if the proportion of accepted proposals falls below 
$\xi_-$, set $\rho^N_{k} = [1-\epsilon] \, \rho^N_{k-1}$ if the proportion of
accepted proposals is above $\xi^+$ and set $\rho^N_{k+1} = \rho^N_k$
otherwise. In other words, the scaling factor is augmented or decreased by an proportion $\epsilon \in (0,1)$ depending on the acceptance rate of the MCMC proposals. In experiments presented in Section \ref{sec:numerical}, we use $\xi_- = 20\%$ and $\xi_+ = 85\%$ and $\epsilon = 20\%$.  \\

\noindent
{\bf Preconditioned Crank-Nicholson Langevin proposals:} one potential drawback of the autoregressive proposals \eqref{eq.pcn} is that no derivative information is exploited. Instead, Markovian proposals $u \mapsto \widehat{u}$ of the type
\begin{align} \label{eq.pcnl}
\widehat{u} = u + (1-\rho) \, \Gamma \, \nabla \log \mu_k(u) + (1-\rho^2)^{1/2} \, \mathcal{N}\BK{ 0, \mathbf{\Gamma} }
\end{align}
can be used within a Metropolis-Hastings scheme for exploring the target density $\mu_k$. Here, $\mathbf{\Gamma}$ is still an approximation of the covariance matrix of $\mu_k$ and $\rho \in (0,1)$ is a scaling factor. Indeed, in the case where the target density is Gaussian, this reduces to the autoregressive proposal \eqref{eq.pcn}. Both the scaling factor $\rho$ and the covariance matrix $\mathbf{\Gamma}$ can be adapted throughout the evolution of a SMC or SET method. In the non-linear-PDE example of Section \ref{sec:BIP}, we describe how gradient/Hessian information can be leveraged to adapt the covariance structure $\Gamma$.

%
%
\subsection{Adaptive number of Mutations}
\label{sec:mutation}
In challenging scenarios, it is important to apply several steps of Markovian mutation at each temperature level. Nevertheless, choosing a sensible number of mutation steps a-priori is often difficult. In this section, we present an adaptive procedure for automatically selecting the appropriate number of mutation steps, inspired by the methodology first proposed in \cite{Kantas2014}. Consider the SET approach when implemented to approximate a target distribution on the state-space $\cX \equiv \mathbb{R}^d$. Furthermore, consider $S \geq 1$ summary statistics, i.e. functions $\summary_s: \cX \to \mathbb{R}$ for $1 \leq s \leq S$.

At iteration $k \geq 0$, right after the resampling step of a SMC
method, or right after the transportation step of the SET approach, consider a set $\{\widetilde{u}^N_{k,i}\}_{i=1}^N$ of particles whose empirical distribution approximates the distribution $\mu_k$. Before applying the mutation kernel $M_k$, the summary statistics are computed, i.e. $\widetilde{s}^{N,s}_{k,i} = \summary_s(\widetilde{u}^N_{k,i})$, for $1 \leq i \leq N$ and $1 \leq s \leq S$. The mutation kernel $M_k$ is then applied to the particles until the correlation along the summary statistics has fallen under a pre-determined threshold $0 < \xi_{\text{stat}} < 1$. In other words, the particles are mutated by defining $\{\widetilde{u}^N_{k,i}[p]\}_{p=0}^{p_k}$ where 
$\widetilde{u}^N_{k,i}[p] \sim M_k\big( \widetilde{u}^N_{k,i}[p-1], d \widehat{u} \big)$, initialized at $\widetilde{u}^N_{k,i}[p=0]=\widetilde{u}^N_{k,i}$, and the number of mutation steps $p_k$ is set as the smallest index $p \geq 1$ such that
\begin{align} \label{eq.corr.adapt}
\textrm{Corr}\BK{ \{\summary_s(\widetilde{u}^N_{k,i}[0]) \}_{i=1}^N,   \{\summary_s(\widetilde{u}^N_{k,i}[p]) \}_{i=1}^N }
\; \leq \; \xi_{\textrm{stat}}
\qquad \textrm{for all} \qquad
1 \leq s \leq S
\end{align}
or when the number of iteration $p_k \geq 1$ reaches a maximum
threshold. The index $p_k$ is referred to as the adaptive {\it number of mutation steps}.
The final mutated particles can thus be described  as
\begin{align*}
u^N_{k,i} \sim M_k^{p_k}\BK{ \widetilde{u}^N_{k,i}, d \widehat{u} }.
\end{align*}
%
A similar approach has been employed in 
\cite{Kantas2014} in which the low-frequencies of a Fourier  expansion
is used
as summary statistics. In Section \ref{sec:BIP}, 
we use as summary statistics the projection of the particles along 
likelihood-informed directions and a threshold of $\xi_{\textrm{stat}} = 80\%$.

%
%
\subsection{Sequential Ensemble Transform: practical implementations}
For completeness, we now described in more details the SET methodology when used in conjunction with the adaptation strategies discussed in Sections \ref{sec:tempmut} and \ref{sec:markovmut}. As in Section \ref{sec:tempmut}, consider an initial distribution $\mu_0$ that is straightforward to sample from, and the posterior distribution $\mu_{\post}$ that can be expressed as $d \mu_{\post} / d \mu_0(u) \propto \exp[V(u)]$ for some potential $V: \mathbb{R}^d \to \mathbb{R}$. We consider tempered distributions $\mu_k$ defined as $d \mu_k / d \mu_0(u) \propto \exp[\tau_k \, V(u)]$ for a temperature parameter $\tau_k \in [0,1]$ that is found adaptively. The entire method is summarised in Algorithm \ref{algo.set}.

\begin{algorithm}[H]
\label{algo.set}
\caption{Sequential Ensemble Transform}
{\bf Inputs:} initial and final distributions $\mu_0$ and $\mu_{\post}$\\
{\bf Output:} particle approximation $(1/N) \sum_{i=1}^N \delta(u^N_{i,\star})$ of the distribution $\mu_{\post}$ 
Set $k=0$ and $\tau_0 = 0$ and initialize $\{u^N_{0,i}\}_{i=1}^N$ as samples from $\mu_0$.\\
\While{$\tau_k < 1$}{
Evaluate $V(u^N_{k,i})$ for $1 \leq i \leq N$.\\
Find the next temperature $\tau_{k+1}$ through Equation \eqref{eq.tau.search}\\
Define the probability weights $w^N_{k+1,i} \propto \exp[(\tau_{k+1} - \tau_k) \, V(u^N_{k,i})]$.\\
Compute the cost matrix $\costmat_{i,j} = \cost(u^N_{k,i}, u^N_{k,j})$\\
Compute $\C^{\OT,N}= \argmin \;C \mapsto \bra{C, \costmat}_{\text{F}} \in \mathbb{R}_+^{N,N}$ under the constraint\\
\begin{align*}
\sum_{j=1}^N \C^{\OT,N}(i,j) = \frac{1}{N}
\qquad \textrm{and} \qquad
\sum_{i=1}^N \C^{\OT,N}(i,j) = w^N_{k+1,j}.
\end{align*} 
Transport the particles by setting: $\widetilde{u}^N_{k+1,i} = N \, \sum_{j=1}^N \C^{\OT,N}(i,j) \, u^N_{k,j}$\\
Use $\{ \widetilde{u}^N_{k+1,i}  \}_{i=1}^N$ to tune a $\mu_{k+1}$-invariant Markov kernel $M_{k+1}(u, d \widehat{u})$. \\
Set $\widetilde{u}^N_{k+1,i}[0] = \widetilde{u}^N_{k+1,i}$ and $p_k=0$.\\
\While{criterion \eqref{eq.corr.adapt} not satisfied}{
Set $p_k \leftarrow p_k+1$\\
Define: $\widetilde{u}^N_{k+1,i}[p_k] \sim M_{k+1}(\widetilde{u}^N_{k+1,i}[p_k-1], d \widehat{u})$
}
Set $u^{N}_{k+1,i} = \widetilde{u}^N_{k+1,i}[p_k]$ and $k \leftarrow k+1$ .
}
\end{algorithm}

%
%
\section{Numerical Experiments}
\label{sec:numerical}
For PDE-constrained Bayesian inverse problems, the overall cost of the
SET algorithm is dominated by PDE solves
\cite{DupontHoffmanJohnsonEtAl03}. Estimating the matrices
$\C^{\OT,N}$ requires solving an optimal transport problem: the standard
simplex method or interior point method \cite{pele2009fast} directly
applied to the linear program \eqref{eq.ot.map} scales as
$\mathcal{O}(N^3)$. Faster and approximate methods are available: for
example, the entropic relaxation of \cite{Cuturi13} computes an
$\epsilon$-approximation with the  cost of $\mathcal{O}(N^2/\epsilon^3)$ \cite{altschuler2017near}.
Our numerical experiments show that even
with 
$N = \mathcal{O}(10^4)$ particles, the computational
overheads associated with solving optimal transport problems to
obtain the optimal transport matrix $\C^{\OT,N}$ is negligible when 
compared to the cost of computing the forward PDE solves. 
Consequently, for all the numerical simulations presented in this
section, the approximate but more scalable methods such as the ones described in
\cite{genevay2016stochastic,Cuturi13} for computing discrete optimal
transport schemes were not employed. Instead, the optimal transport matrices
were computed through a standard 
simplex solver \cite{flamary2017pot}.
To operate, the SET method requires $\mathcal{O}\BK{ N \times M}$
PDE-solves where $M$ is the total number of Markov mutations applied
to each particle. In this section, we adopt the strategies described
in Sections \ref{sec:tempmut} and \ref{sec:markovmut} and
\ref{sec:mutation} for automatically adapting the sequence of
temperatures, the Markov mutation kernels, and the number of times
these Markov mutation kernels were applied.
We compare the SET approach to the state-of-the-art adaptive SMC 
approach of \cite{Kantas2014, Beskos2015a}.
In this section, we present three numerical experiments with
increasing complexity. The first experiment investigates the influence
of the mixing properties of the mutation kernels: for this purpose,
the adaptive schemes used for adapting the Markov kernels, temperature
ladder, and number of mutations at each temperature are switched
off. 
The second experiment looks into the effect of the number of
mutations at each temperature level. Finally, the last experiment is a
relatively challenging Bayesian inverse problem. It illustrates the
robustness and efficiency of the SET method when used in conjunction
with automated adaptation strategies; to the best of our knowledge,
the scheme using the averaged Gauss-Newton Hessian for adapting the 
PCNL covariance structure is new.

\subsection{Scalar Target Distribution}
\label{sec:one:d}
In this section, we investigate the influence of the mixing properties of the Markov mutation kernels.
We consider a one-dimensional Gaussian target distribution $\mu(du)$ defined as 
\begin{align} \label{eq.experiemtn.scalar}
\frac{d \mu }{ d \mu_0 }(u) \propto \exp \curBK{-\frac{1}{\sigma_{\textrm{noise}}^2}\BK{u - 1/2}^2} 
\equiv \exp\sqBK{V(u)}
\end{align}
for a ``prior" distribution $\mu_0$ chosen as a centred Gaussian with
unit variance $\sigma_0=1$. In the experiments presented in this
section, we chose $\sigma_{\textrm{noise}} = 10^{-3}$: although all
the quantities are Gaussian, this setting is challenging since
$\sigma_{\textrm{noise}} \ll \sigma_0$. In order to focus on the
mixing properties of the Markov mutation kernels, we fix a sequence of
intermediate temperatures equally spaced on a logarithmic scale $\{
\tau_k \}_{k=1}^K$ with $K=30$. In other words, the adaptive tempering
scheme presented in Section \ref{sec:tempmut} is not used. Denote by
$\sigma_k$ the standard deviation of the Gaussian intermediate
distribution $\mu_k$ defined as $d \mu_k / d \mu_0(u) \propto
\exp[\tau_k \, V(u)]$. At temperature $\tau_k > 0$, the Markov
mutation kernel is chosen as a Random Walk Metropolis (RWM) kernel
with Gaussian perturbations with standard deviation $\rho \,
\sigma_k$, where $\rho > 0$ is used to control the mixing properties of the mutation kernels. For $\rho \ll 1$, the mutation kernels are inefficient while for $\rho \approx 1$ the mutation kernels are close to optimal.\\

%
%
\begin{figure}[h!t!b!]
\begin{center}
\subfigure{
\includegraphics[width=1\textwidth]{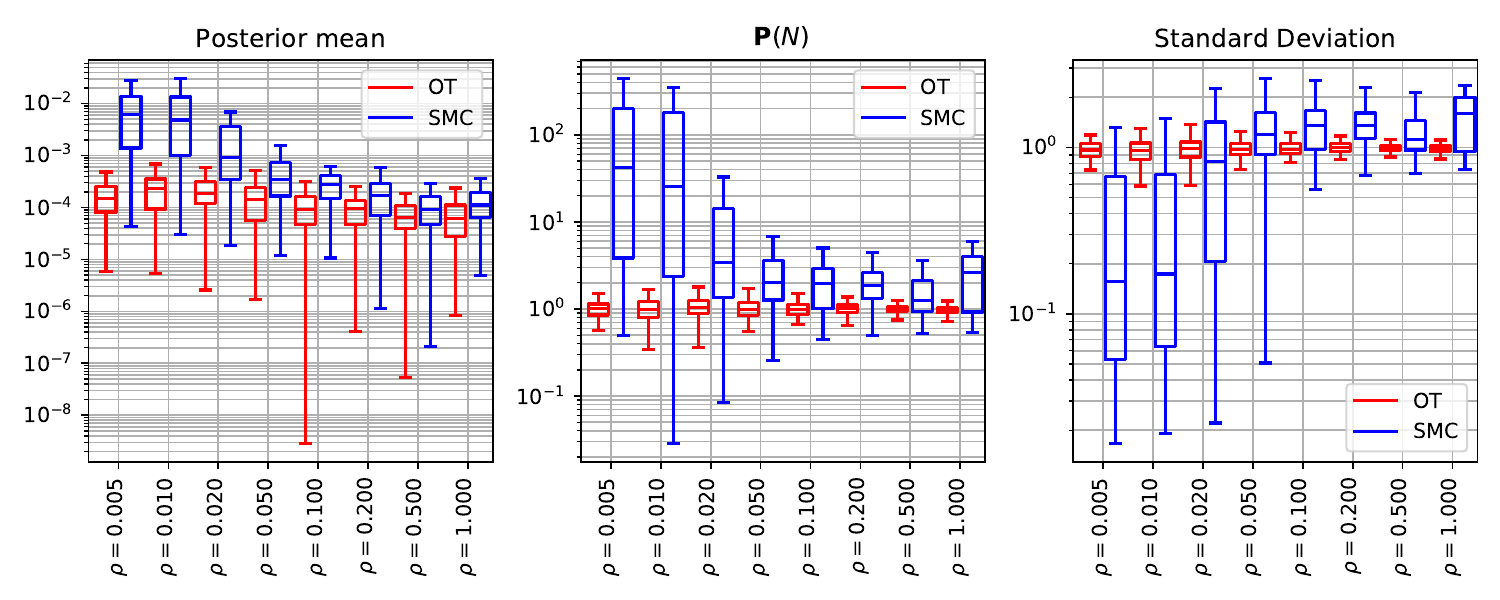}
}
\end{center}
\caption{Target distribution \eqref{eq.experiemtn.scalar} with
  $N=10^2$ particles, no adaptation and a ladder of $K=30$ temperatures 
  equally spaced on a logarithmic scale. Each experiment is executed
  and averaged over $n=100$ times. The scaling parameter $\rho>0$ quantifies the quality of the Markovian mutations.
{\bf Left:} distribution of $|\widehat{m}_{\post}^N - m_\post|$ 
{\bf Middle:} distribution of the quantity $\mathbf{P}(N)$ defined in \eqref{eq.nll} 
{\bf Right:} distribution of the ratio $\widehat{\sigma}_{\post}^N / \sigma_\post$}
\label{fig:one_dim}
\end{figure}

Set $m_{\post}$ and $\sigma_\post$ the mean and standard deviation of
the target distribution $\mu(du)$. For a particle approximation $\mu^N
= (1/N) \sum_{i=1}^N \delta(u^N_i)$, we take $\widehat{m}_{\post}^N$
and $\widehat{\sigma}_{\post}^N$ as its mean and standard
deviation. We also consider the quantity $\mathbf{P}(N)$ that equals,
up to irrelevant additive and multiplicative constants, the negative log-posterior, 
\begin{align} \label{eq.nll}
\mathbf{P}(N) \equiv 
\frac{1}{N}\, \sum_{i=1}^N \frac{ (u^N_i - m_{\post})^2 }{ \sigma^2_\post}.
\end{align}
In this Gaussian setting and in the idealized situation when the
samples $\{u^N_i\}_{i=1}^N$ are i.i.d samples from $\mu_\post$ and $N
\to \infty$, the quantity \eqref{eq.nll} converges to one. Figure
\ref{fig:one_dim} reports the quality of the approximation of the
mean, standard deviation, and the quantity \eqref{eq.nll}
when the SET and SMC methods are employed with $N = 10^2$ particles
and identical conditions. For each value of $\rho$, the same
experiment is executed $n=100$ times. For quantifying the quality of
the approximation of the posterior mean, the absolute difference
$|\widehat{m}_{\post}^N - m_\post|$ is reported. For quantifying the
approximation of the standard deviation, the ratio
$\widehat{\sigma}_{\post}^N / \sigma_\post$ is reported. Finally, 
the quantity \eqref{eq.nll} is also reported:
values closer to one indicate a better calibrated approximation. In
this setting, the SET approach outperforms the SMC method over all the metrics. Furthermore and as expected, as $\rho \to 0$, i.e. as the mixing of the mutation kernels gets worse, the efficiency of the SMC approach degrades. 
Although the theoretical results described in Section \ref{sec.consistency} does not explain this phenomenon, 
the SET method appears to continue to perform well in the regime $\rho \to 0$ in that example.

\subsection{Multivariate Gaussian Target Distribution}
\label{sec.toy}
%
%
\begin{figure}[h!t!b!]
\begin{center}
\subfigure{
\includegraphics[width=1\textwidth]{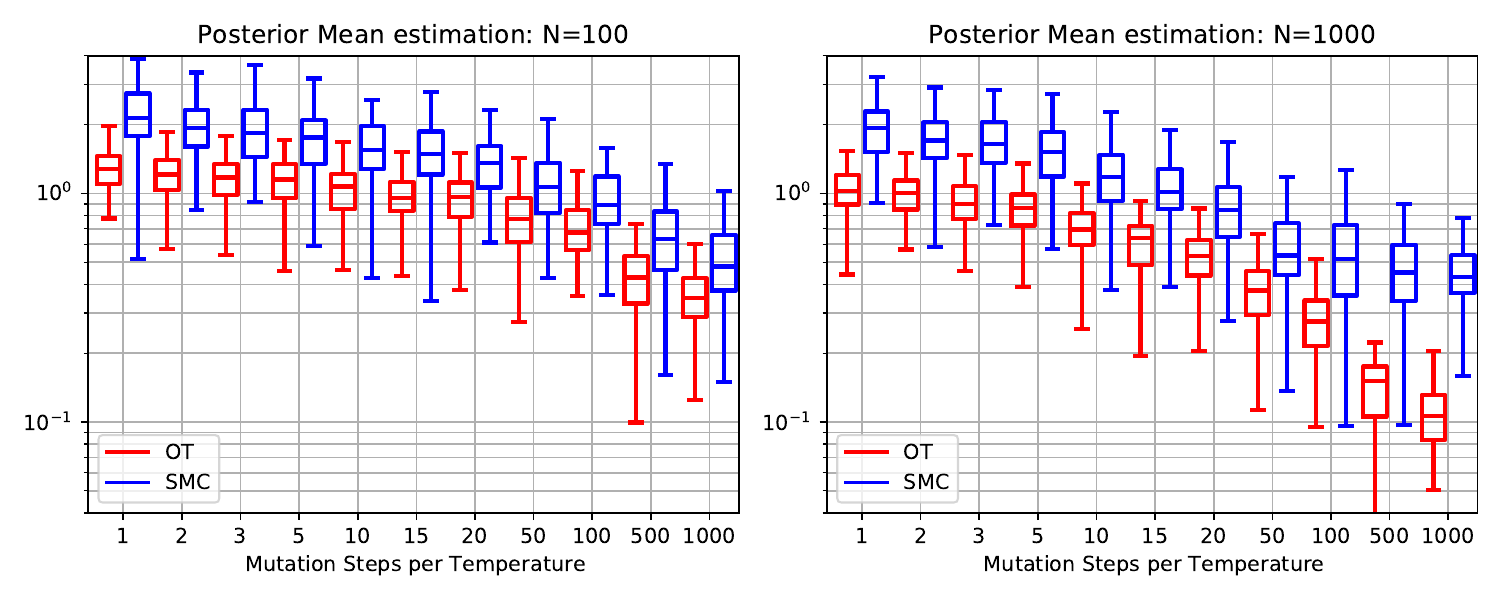}
}
\end{center}
\caption{Estimation of the posterior mean of distribution
\eqref{eq.numerical.mult.gauss} with $N=10^2$ {\bf (left)} and
$N=10^3$ {\bf(right)} particles. The error 
$\|\widehat{m}_{\post}^N- m_\post\|$ is plotted against the number of mutation 
steps $p \geq 1$  at each temperature level. Each experiment is averaged over $n=50$ runs.}
\label{fig:nondiagdim20mean}
\end{figure}

%
%
\begin{figure}[h!t!b!]
\begin{center}
\subfigure{
\includegraphics[width=1\textwidth]{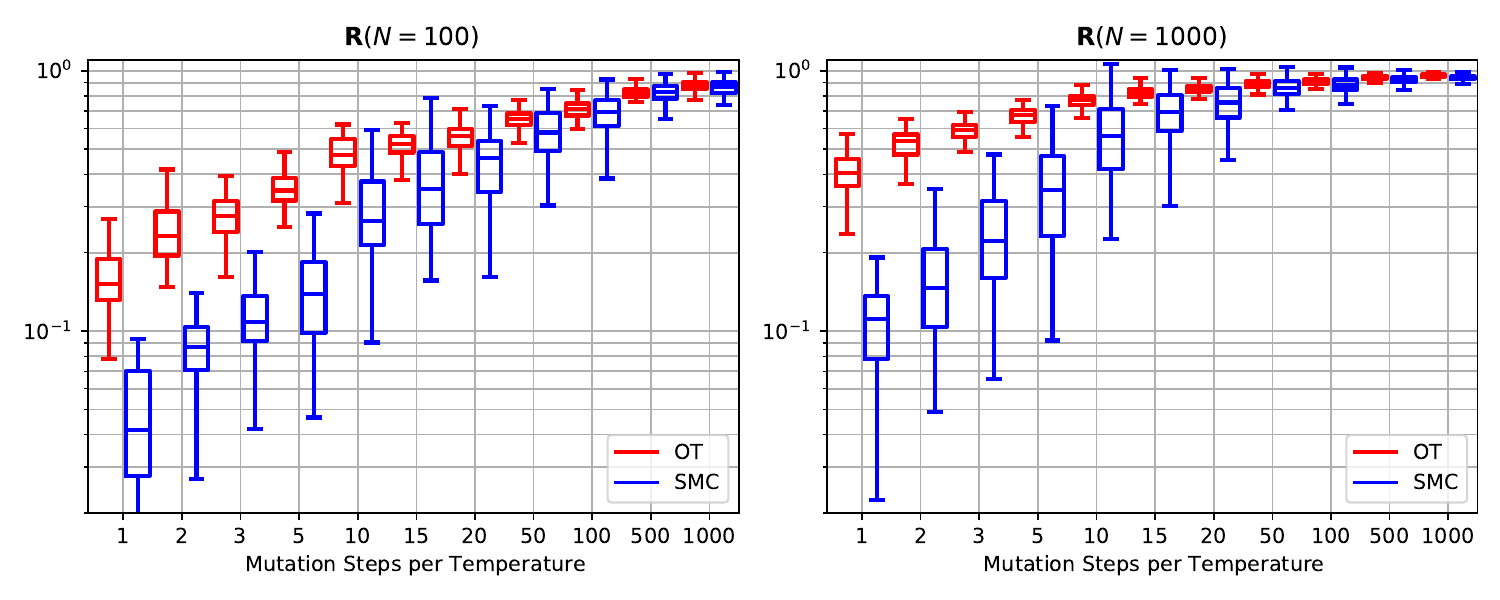}
}
\end{center}
\caption{Estimation of the posterior standard deviation of distribution \eqref{eq.numerical.mult.gauss} with $N=10^2$ {\bf (left)} and $N=10^3$ {\bf(right)} particles. The quantity $\mathbf{R}(N)$ defined in Equation \eqref{eq.error.variance} is  plotted against the number $p \geq 1$ of mutation steps at each temperature level. Each experiment is repeated $n=50$ times.}
\label{fig:nondiagdim20std}
\end{figure}

In this section, we study the influence of the number of mutation
steps at each temperature in a more challenging scenario. As opposed to Section \ref{sec:one:d}, we employ the adaptive tempering scheme described in Section \ref{sec:tempmut}.
Let $\mu_0$ be a centered Gaussian distribution in $\mathbb{R}^D$ with identity covariance matrix.
The Gaussian target distribution $\mu$ is defined through the change
of probability measure
\begin{align} \label{eq.numerical.mult.gauss}
\frac{d \mu}{d \mu_0}(u) \; \propto \;
\exp \curBK{-\frac12 \, \bra{u, \Gamma^{-1} u} }.
\end{align}
The covariance matrix $\Gamma \in \mathbb{R}^{D,D}$ is given by
\begin{align*}
\Gamma_{i,j} = \sigma^2 \, \exp \curBK{- \frac{(j-i)^2}{2 \, \ell^2}}
\end{align*}
for a variance parameter $\sigma^2>0$ and length-scale parameter $\ell > 0$. In the
numerical experiments of this section, we chose $\sigma=1$ and $\ell = 4$ and $D=20$.  
Although the target distribution is Gaussian, it is a challenging scenario since it is already 
relatively high-dimensional ($D=20$) and the covariance matrix of the posterior distribution 
is highly ill-conditioned (and hence the posterior probability mass is
concentrated in low dimensional spaces dictated by dominant
eigenvectors of the covariance matrix).
The SET and SMC methods have been implemented with a number of particles 
$N \in \{10^2, 10^3\}$ and an effective sample size 
threshold \eqref{eq.tau.search} is set to $\xi_\ess = 1/2$. Furthermore, we used 
autoregressive MCMC proposals as defined in Equation \ref{eq.pcn} with mean 
and covariance structure empirically estimated from the population of
particles. In particular,
the covariance matrix of the autoregressive proposals is assumed to be diagonal 
with empirical marginal variances on the diagonal (see \cite{Kantas2014} for a
similar approach in the SMC context). As described in Section \ref{sec:tempmut}
the scaling parameter $\rho > 0$ was chosen adaptively to maintain MCMC mutations 
with acceptance rates in between the thresholds $\xi_{-}=20\%$ and $\xi_{+}=80\%$.

%
%
\begin{figure}[h!t!b!]
\begin{center}
\subfigure{
\includegraphics[width=0.33\textwidth]{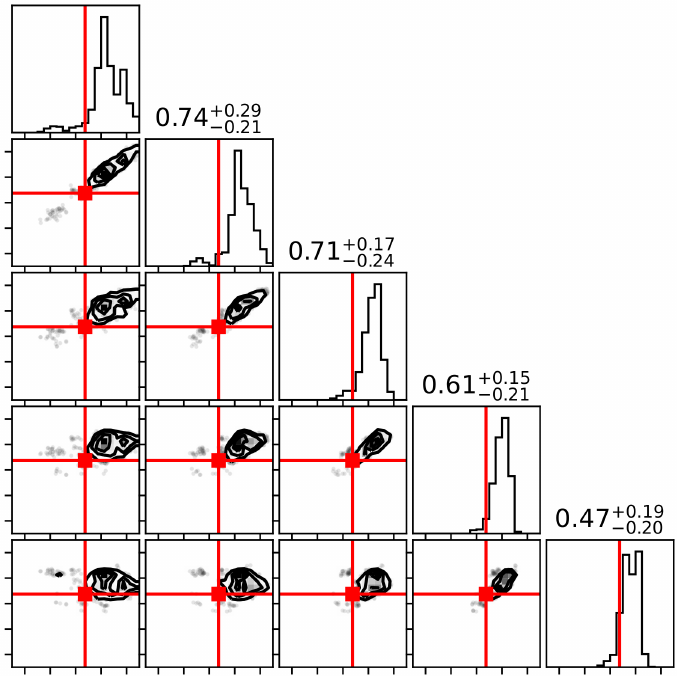}
\includegraphics[width=0.33\textwidth]{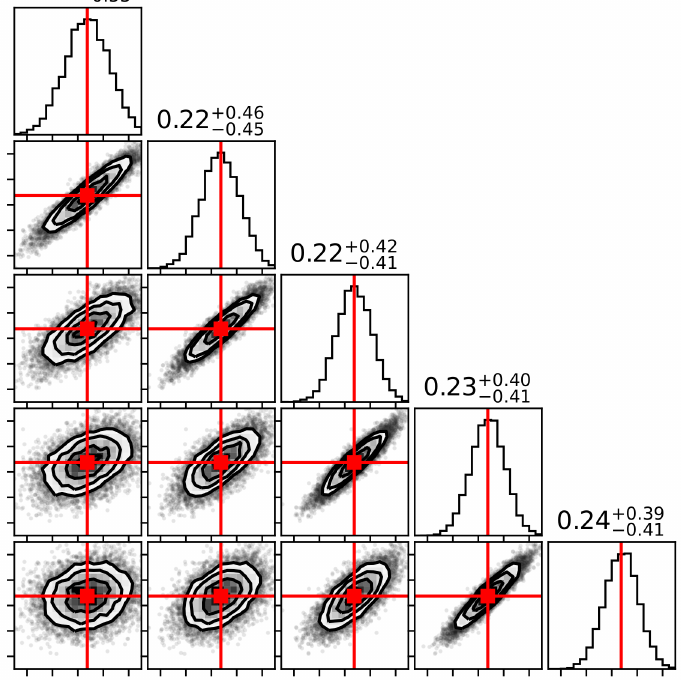}
\includegraphics[width=0.33\textwidth]{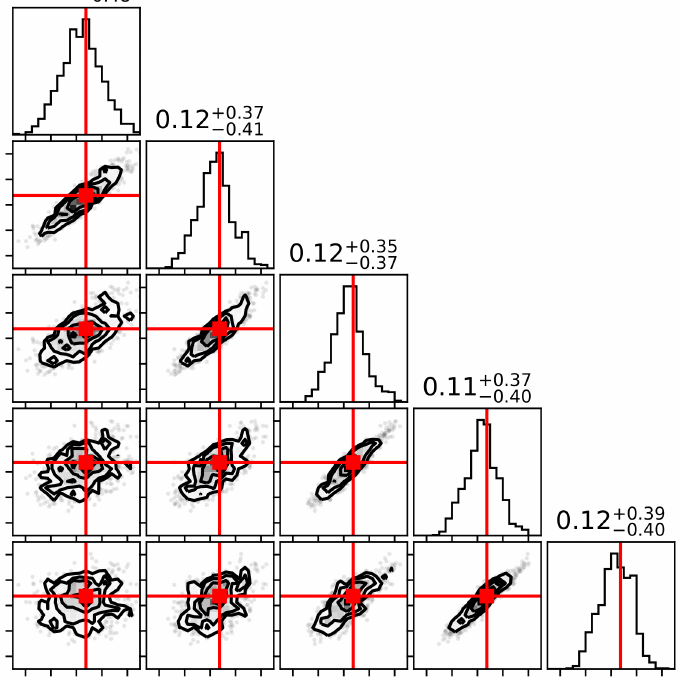}
}
\end{center}
\caption{Density plots of the four first coordinates of the $20$-dimensional target distribution \eqref{eq.numerical.mult.gauss} for SMC {\bf (left)} and SET {\bf (right)}  with $N=10^4$ particles 
and $p=20$ mutations at each temperature. {\bf Middle} is the density plot of $N=10^4$ independent samples from the target distribution.}
\label{fig:temp_pair}
\end{figure}

We compared the performance of the SET and SMC methods when used with 
a fixed number $p \geq 1$ of mutation steps at each temperature level. 
Experiments were carried out for a number of mutations as low as $p=1$ and 
as high as $p=10^3$.
As in Section \ref{sec:one:d}, 
we report the quality of the posterior mean and posterior standard deviation. 
In Figure \ref{fig:nondiagdim20mean}, for each value of the number of mutation steps $p \geq 1$, 
$\|\widehat{m}^N_{\post} - m_\post\|$ is averaged over  $50$ runs, 
where $\widehat{m}_\post^N \in \mathbb{R}^D$ denotes the posterior mean estimated 
from the population of particles. Similarly, 
in Figure \ref{fig:nondiagdim20std}
we report the averaged ratio $\mathbf{R}(N)$ between the estimated standard 
deviations and the theoretical ones,
\begin{align} \label{eq.error.variance}
\mathbf{R}(N)
\equiv \frac{1}{D} \sum_{d=1}^D
\frac{\widehat{\sigma}^N_{\post, d}}{\sigma_{\post,d}},
\end{align}
where $\sigma_{\post,d}$ is the marginal standard deviation in the
$d$-th dimension 
and $\widehat{\sigma}^N_{\post, d}$ is its estimate obtained from the population of particles.
Similarly to Section \ref{sec:one:d}, we observe that the SET method 
appears to be more robust when the number of mutations $p$ is very low. 
As $p$ increases, the difference between the two methods progressively 
disappears and $\mathbf{R}(N) \to 1$ for both methods.

\begin{figure}[h!t!b!]
\begin{center}
\subfigure{
\includegraphics[width=1\textwidth]{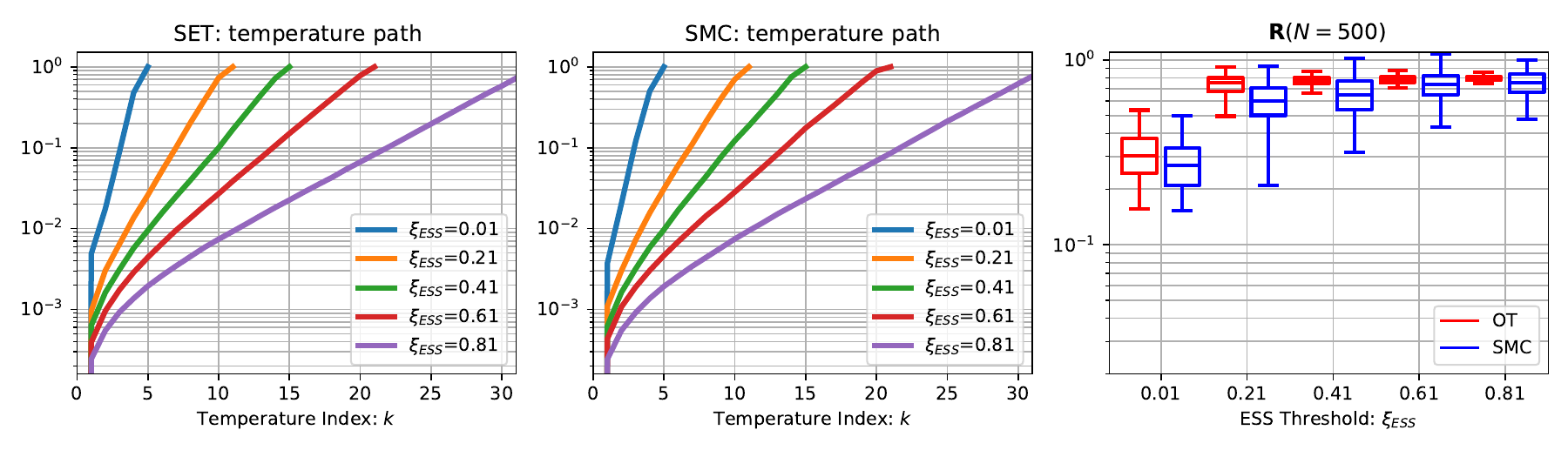}
}
\end{center}
\caption{Temperature trajectories 
  for 
  the target distribution \eqref{eq.numerical.mult.gauss} using the
  SET {\bf (left)} and the SMC {\bf (center)} methods
  with $N=500$ particles and $p=20$ mutations at each temperature: each trajectory corresponds to a different effective sample size threshold $\xi_\ess \in \{ 1\%, 21\%, 41\%, 61\%, 81\% \}$. The plot on the {\bf (right)} displays the the $\mathbf{R}(N)$ statistics, which quantifies the approximation of the posterior mean, for each value of $\xi_{\ess}$.}
\label{fig:temp_path}
\end{figure}

Figure \ref{fig:temp_pair} shows the result 
with $N=10^4$ particles 
and $p=20$ mutations for each temperature. The marginal pairwise  distribution of the first 
four dimensions are displayed: although for $p=20$ neither SMC nor SET 
produces an entirely satisfactory approximation of the target distribution, 
it is qualitatively visible that the SET produces an approximation that is closer to the 
correct distribution.

Finally, in order to gain some understanding of the influence of the effective sample size 
threshold $\xi_\ess$ on the sequence of temperatures, as well as to study the sequence 
of temperatures adaptively chosen by the SMC and SET methods, Figure \ref{fig:temp_path} 
displays the temperature paths for the SMC and SET methods. 
As expected, larger values of the effective sample size threshold $\xi_\ess$ lead to a slower increase 
of the (inverse) temperature parameter. 
Furthermore, low values of the effective sample size threshold results in a loss of accuracy. 
This phenomenon is well understood for SMC methods since lowering $\xi_\ess$ exacerbates 
particle degeneracy 
but more investigations are required for SET to understand in more details 
the mechanisms. 
Figure \ref{fig:temp_path} also shows that, except at the very start of 
the algorithm, the (inverse) temperature increases roughly linearly on a logarithmic scale. 
Furthermore, when the number of mutations  per temperature is fixed, as is done in this 
example, the SET and SMC temperature trajectories are very close to each other. 
This remark is important since it ensures that the numerical results presented in 
Figures \ref{fig:nondiagdim20mean} and \ref{fig:nondiagdim20std} are
fair, that is, for each number $p \geq 1$ mutations per temperature, the computational budgets used by the SMC
 and SET methods 
are equivalent.

\subsection{Bayesian Inverse Problem}
\label{sec:BIP}
In this section, we test the SET method for inference in a Bayesian inverse problem 
governed by a Partial Differential Equation (PDE). More specifically, we consider 
the following Poisson PDE on the unit disk $\Omega \subset \mathbb{R}^2$,
\begin{align} \label{eq.elliptic.pde}
\nabla \cdot \BK{e^z \, \nabla f}(x) = h(x)
\qquad \textrm{for} \qquad x \in \Omega,
\end{align}
for a known source term $h:\Omega \to \mathbb{R}$ , Dirichlet boundary 
conditions $f(x) = 0$ for $x \in \partial \Omega$ and a temperature field $f:\Omega \to \mathbb{R}$. 
We are interested in reconstructing $z:\Omega \to \mathbb{R}$,
the log-conductivity field, from 
noisy observations collected at locations $x_1, \ldots, x_K \in \Omega$ modelled 
as $d_i = f(x_i) + \eta_i \in \mathbb{R}$ with independent Gaussian random 
noises $\eta_1, \ldots, \eta_K$ centered at $0$ and with variance $\sigma_{\noise}^2$. 
We assume a Gaussian prior distribution $\mu^{z}_\prior$ on the log-conductivity 
field $z:\Omega \to \mathbb{R}$ with Matern covariance structure \cite{LindgrenRueLindstroem11}. 
Draws from this prior distribution can be generated by solving the elliptic PDE
\begin{align} \label{eq.pde.prior}
\BK{\kappa^2 - \Delta} z(x) = w(x) \qquad \textrm{for} \qquad x \in \Omega
\end{align}
with vanishing Dirichlet boundary conditions. The right-hand-side of Equation \eqref{eq.pde.prior} 
is the realization $w: \Omega \to \mathbb{R}$ of a 
Gaussian white noise process on the domain $\Omega$. 
For the numerical experiments, the scale parameter $\kappa$ was set to $ 10^{-1}$. The posterior 
distribution $\mu^{z}_\post(dz)$ on the log-conductivity fields reads
\begin{align} \label{eq.posterior.pde}
\frac{d \mu^{z}_\post}{d \mu^{z}_\prior}(z) \; \propto \; 
\exp \curBK{-\frac{1}{2 \, \sigma^2_\noise} \sum_{i=1}^K \BK{d_i - \forward(z)[x_i]}^2 }
\end{align}
where $\forward$ is the parameter-to-observable map 
that associates to each log-conductivity field $z$ 
the corresponding temperature field $f = \forward(z)$ obtained by solving the PDE \eqref{eq.elliptic.pde}. 
In our experiments, we assume that the standard deviation of the additive Gaussian 
noise is known, i.e. $\sigma_\noise = 10^{-2}$. 
The location of the observations, the ground truth temperature 
and log-conductivity fields are depicted in Figure
\ref{fig:IP_observations}. Here, the ground truth 
log-conductivity field was obtained as a draw from the prior distribution \eqref{eq.pde.prior}.

\begin{figure}[h!t!b!]
\centering
\subfigure{
\includegraphics[width=1\textwidth]{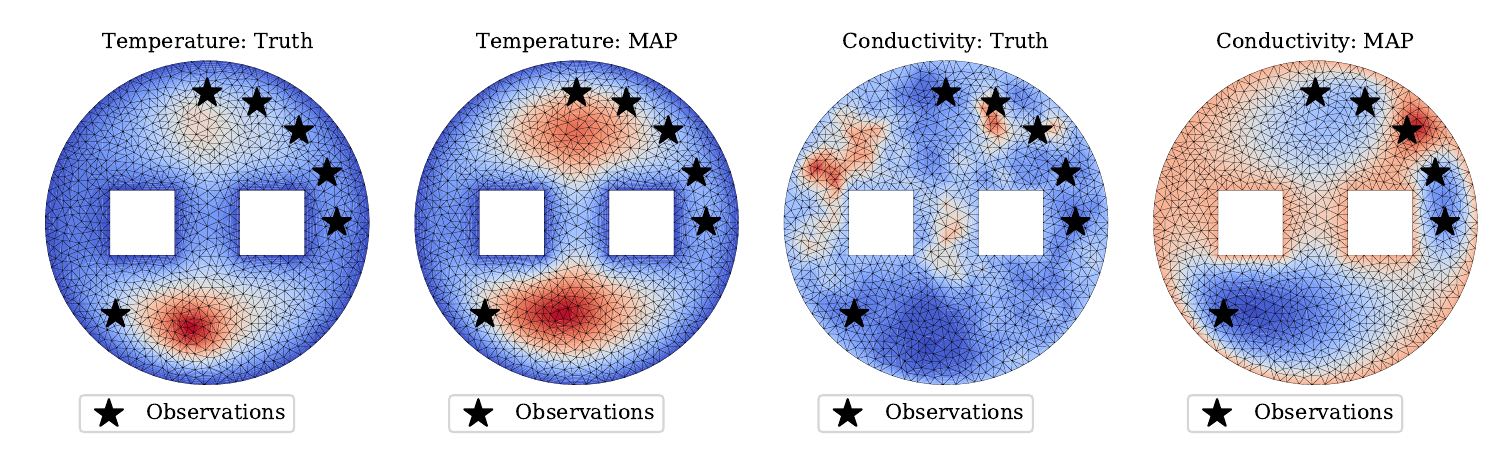}
}
\caption{Ground truth and MAP estimates of the temperature field $f:\Omega \to \mathbb{R}$ and log-conductivity field $z:\Omega \to \mathbb{R}$ in the inverse problem \eqref{eq.elliptic.pde}. The prior distribution on the log-conductivity field is a Gaussian field with vanishing mean and Mattern covariance structure \eqref{eq.pde.prior}.}
\label{fig:IP_observations}
\end{figure}
\subsubsection{Discretization and parametrization}
The PDE \eqref{eq.elliptic.pde} was discretized with the Finite Element Method (FEM)
implemented on a mesh $\mesh$ with $M=1170$ nodes, as shown in Figure \ref{fig:IP_observations}, 
using {\it FEniCS} \cite{DupontHoffmanJohnsonEtAl03}. In the remainder of this section, we 
consequently approximated all functions defined on the domain $\Omega$ with their projection 
on the finite element function space with piecewise linear basis functions $e_i:\Omega \to \mathbb{R}$ 
for $1 \leq i \leq M$. In other words, the function space (infinite dimensional) Bayesian posterior described in Equation \eqref{eq.posterior.pde} is approximated by a $M$-dimensional posterior distribution. 
To avoid the notational 
burden, we use the same notations to refer to the original (infinite
dimensional) quantities and their FEM 
approximations. 
Similarly, the notation $\forward$ refers to both the original forward operator and its discretized version.
If $\mass$ denotes the mass matrix associated to the FEM basis $\{e_i\}_{i=1}^M$, the discretization of a Gaussian white noise on $\Omega$ can be realized as 
$w(x) = w_1 \, e_1(x) + \ldots + w_M \, e_M(x)$ where $\ww = (w_1, \ldots, w_M) \in \mathbb{R}^M$ is a realization of a centered Gaussian random variable with covariance matrix $\mass^{-1}$. We used a (sparse) 
Cholesky decomposition $\mass = \mathbf{L} \mathbf{L}^\top$ and expressed the white noise vector 
as solution of the linear system $L^\top \, \ww = \uu$ where $\uu \in \mathbb{R}^M$ is the realization of a centered standard Gaussian distribution with identity covariance matrix. For convenience, denote by $\Phi: \mathbb{R}^M \to \mathbb{R}^M$ the (linear) 
operator that maps $\uu$ to the corresponding log-conductivity field. In other words, the function $z(x) = z_1 \, e_1(x) + \ldots + z_M \, e_M(x)$, with $(z_1, \ldots, z_M) = \zz = \Phi(\uu) \in \mathbb{R}^M$, is the solution obtained through the FEM of the PDE \eqref{eq.pde.prior} with right-hand-side represented by $\ww = (L^\top)^{-1} \, \uu$. When implementing the SET and SMC methods, the log-conductivity field is parametrized through the quantity $\uu$. This parametrization has the advantage of corresponding to a standard isotropic Gaussian prior distribution with identity covariance matrix and a posterior distribution that is close to a standard Gaussian distribution except along a few data-informed directions \cite{Flath2011,Bui-ThanhBursteddeGhattasEtAl12,Bui-Thanh2013,CuiLawMarzouk16}. These properties lead to Markov mutation kernels that are easier to tune and automatically adapted. In this parametrization, the $\mathbb{R}^M$-valued posterior density $\mu^{\uu}_\post(\uu)$ reads
\begin{align} \label{eq.posterior.pde.discrete}
 \mu_\post^{\uu}(\uu) \propto \exp\curBK{-\frac{1}{2} \|\uu\|^2 -\frac{1}{2 \, \sigma^2_\noise} \sum_{i=1}^K \BK{d_i - \forward(z)[x_i]}^2} 
\propto \mu^{\uu}_0(\uu) \, \exp\sqBK{V(\uu)}
\end{align}
where $\mu^{\uu}_0$ is the density of a centered standard isotropic Gaussian distribution in $\mathbb{R}^M$ and $V(\uu) = -(1/2)\sigma^{-2}_{\noise} \sum_{i=1}^K \BK{d_i - \forward(z)[x_i]}^2$ is the negative of the {\it data-misfit} functional.\\

For implementing the SET method, a cost matrix is needed. In order to take into account the geometry of the problem, when the SET method is implemented with $N$ particles $\{\uu^N_i\}_{i=1}^N$, the costs matrix $\costmat \in \mathbb{R}^{N,N}_+$ is defined as follows. The entry $\costmat_{i,j}$ is set to the squared $L^2(\Omega)$ distance between the log-conductivity fields associated to the particles $\uu_i$ and $\uu_j$,
\begin{align} \label{eq.dist.pde}
\costmat_{i,j} = \bra{\uu^N_j, \mass \, \uu^N_j}.
\end{align}
\subsubsection{Adaptive scheme}
\label{subsec:adaptive.scheme.BIP}

\begin{figure}[h!t!b!] \label{fig.directions.BIP}
\centering
\subfigure{
\includegraphics[width=1\textwidth]{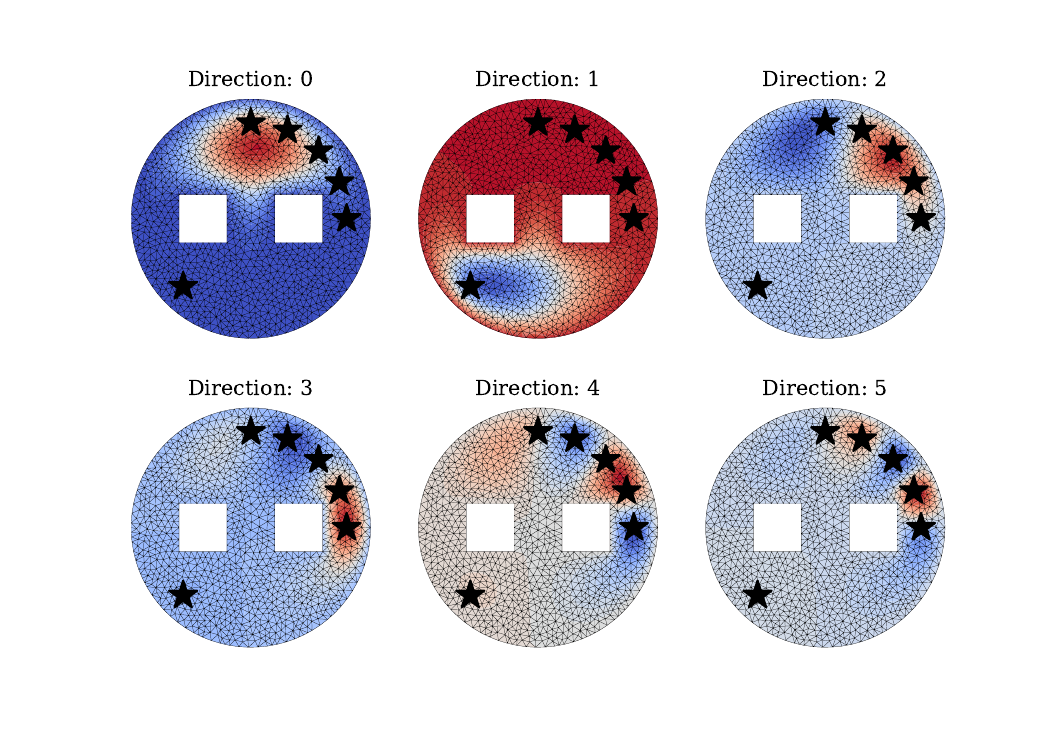}
}
\caption{Directions $\{\Phi(\vv_k)\}_{k=1}^K$ associated to the $K=6$ dominating eigenvectors $\{\vv_k\}_{k=1}^K$ of the Gauss-Newton Hessian $\hess_{\textrm{GN}}(\uu_{\textrm{MAP}})$ defined in Equation \eqref{eq.GN.hess}.}
\end{figure}

To automate the choice of the number of mutation steps at each
temperature, the adaptive scheme presented in Section
\ref{sec:mutation} is used. In order to obtain meaningful summary
statistics, we considered a data-driven approach. First, the {\em
  maximum a posterior} (MAP) estimate $\uu_{\textrm{MAP}}$ was
obtained by minimizing the negative log-posterior density $\uu \mapsto
- \log \mu^\uu_{\post}(\uu)$: gradients were computed with the adjoint
method and a standard L-BFGS minimization procedure was used. Figure
\ref{fig:IP_observations} displays the MAP estimate as well as the ground
truth. Then, we formed a Gauss-Newton approximation\footnote{Note that
the Gauss-Newton approximation includes a term corresponding to the Gaussian prior.} $\hess_{\textrm{GN}}(\uu_{\textrm{MAP}})$ of the Hessian to $-\log \mu^\uu_{\post}$ at the MAP,
\begin{align} \label{eq.GN.hess}
\hess_{\textrm{GN}}(\uu_{\textrm{MAP}}) = \mathbf{I} + \frac{1}{\sigma^2_\noise} \sum_{i=1}^K 
\BK{ \nabla_\uu \forward(z)[x_i]} \BK{ \nabla_\uu \forward(z)[x_i]}^\top \in \mathbb{R}^{M,M}.
\end{align}
On the right-hand-side of \eqref{eq.GN.hess}, all the gradient terms are evaluated at $\uu = \uu_{\textrm{MAP}}$.
The eigenvectors $\vv_1, \ldots, \vv_K \in \mathbb{R}^M$ corresponding
to the $K$ dominating eigenvalues of the Gauss-Newton Hessian
$\hess_{\textrm{GN}}(\uu_{\textrm{MAP}})$ span the directions along
which the collected data are the most informative and along
which the posterior distribution differs most from the prior distribution 
\cite{BashirWillcoxGhattasEtAl08,Bui-Thanh2012}.
Finally, we chose $S=K$ summary statistics defined as $\summary_k(\uu)
\equiv \bra{\vv_k, u}$. Figure \ref{fig.directions.BIP} shows the $K =
6$ directions $\Phi(\vv_k)$ for $1 \leq k \leq K$.

We used Preconditioned Crank-Nicholson Langevin (PCNL) proposals, as
described in Section \ref{sec:markovmut}, for mutating the particles:
the scaling parameter $\rho \in (0,1)$ was adapted so as to maintain
an acceptance probability in between $\xi_- = 20\%$ and $\xi_+ =
80\%$. The structure of the covariance  $\Gamma$ of the PCNL proposals
was also chosen adaptively. When exploring the density $\mu_k(\uu)
\propto \mu^{\uu}_0(\uu) \, \exp[\tau_k \, V(\uu)]$ at temperature $\tau_k$, the particle system $\{\uu^N_i\}_{i=1}^N$ was used to approximate the averaged Gauss-Newton Hessian $(1/N) \sum_{i=1}^N \hess^{\tau_k}_{\textrm{GN}}(\uu^N_i)$ where
\begin{align} \label{eq.avg.hessian}
\hess^{\tau_k}_{\textrm{GN}}(\uu) = 
\mathbf{I} + \frac{\tau_k}{\sigma^2_\noise} \sum_{i=1}^K 
\BK{ \nabla_\uu \forward(z)[x_i]} \BK{ \nabla_\uu \forward(z)[x_i]}^\top \in \mathbb{R}^{M,M}.
\end{align}
In the setting when the prior distribution is Gaussian and the forward map is linear, 
the posterior distribution is also Gaussian and the averaged
Gauss-Newton Hessian equals the precision of this Gaussian
posterior distribution. This motivates the use of the averaged Gauss-Newton
Hessian\textemdash which is guaranteed to be positive definite\textemdash as the inverse covariance structure for the noise used within the PCNL proposals. To summarize, at temperature $\tau_k$ and right after the transportation step when using the SET approach, or right after the resampling step when using SMC, the covariance $\widehat{\Gamma}_k$ used within the PCNL proposals was defined as
\begin{align} \label{eq.adapt.PCNL}
\widehat{\Gamma}_k = \curBK{\frac{1}{N} \, \sum_{i=1}^N \, \hess^{\tau_k}_{\textrm{GN}}(\uu^N_i) }^{-1}
\end{align}
where $\{\uu^N_i\}_{i=1}^N$, again, denotes the current particle population.

\subsubsection{Results}
\label{subsec:results.BIP}

\begin{figure}[h!t!b!]
\begin{center}
\subfigure{
\includegraphics[width=1\textwidth]{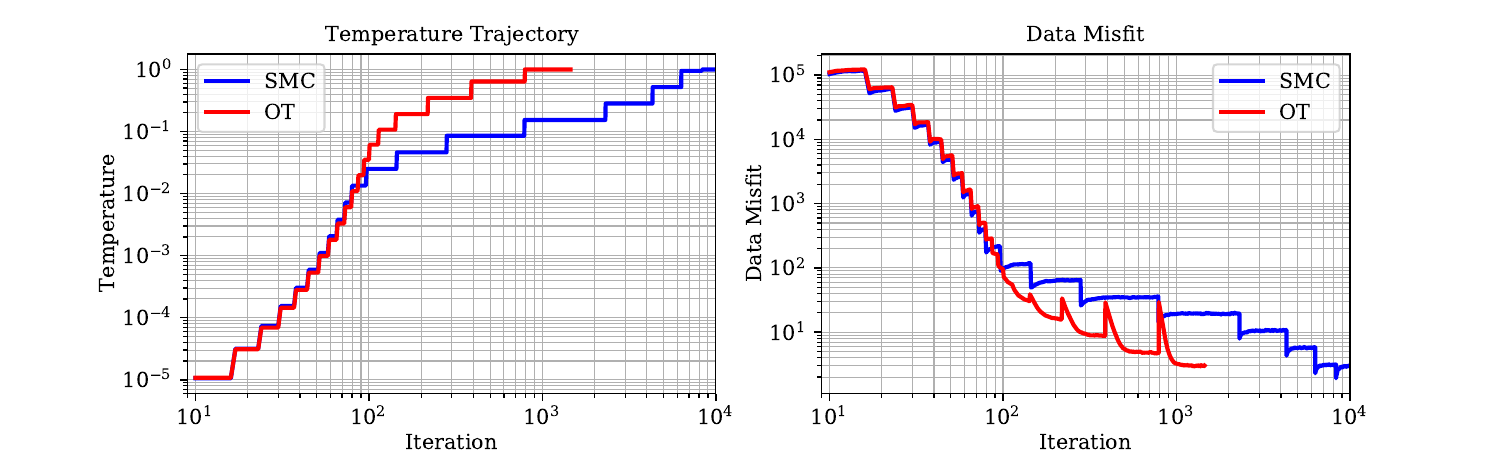}
}
\end{center}
\caption{Bayesian Inverse Problem \eqref{eq.elliptic.pde}:
  Trajectories for Temperature {\bf (left)} and averaged data-misfit functional {\bf (right)} using the SET and SMC methods with $N=2.10^3$ particles, adaptive PCNL mutation kernels, an adaptive temperature scheme, and an adaptive number of mutation steps at each temperature.}
\label{fig:IP_temp_traj}
\end{figure}

We implemented the SET and SMC approaches with $N=2.10^3$ particles
with identical conditions on a server with
$20$ computing cores, one for each particle to be computed in parallel. 
The initial distribution was chosen as the prior, i.e. $\mu^\uu_0(\uu) \propto \exp\sqBK{-(1/2)\|\uu\|^2}$.
Furthermore, the schemes presented in Sections  \ref{sec:tempmut}, \ref{sec:markovmut} and \ref{sec:mutation} were used to automatically adapt the ladder of temperatures, the Markov mutation kernels, and the number of  mutations.

The ground truth was obtained by running $20$ Preconditioned
Crank-Nicholson Langevin MCMC simulations (in parallel) 
for $L=10^7$ iterations, each of the runs was initialized from independent draws using the Gauss-Newton approximation described in Section \ref{subsec:adaptive.scheme.BIP}. Convergence was checked by verifying that the marginal means, variances, and summary statistics described in Section \ref{subsec:adaptive.scheme.BIP} agreed among the $20$ chains.

Figure \ref{fig:IP_temp_traj} (left) shows the trajectories of the
temperatures for both the SET and SMC methods as a function of the
total number of mutation steps. In this example, the Markov mutation
steps are more computationally expensive, by at least an order of
magnitude, than all the other computational overheads. Consequently,
the total number of mutation steps is roughly proportional to the
wall-clock compute time. Note that the SET method requires almost an
order magnitude less iterations to converge
since, as numerically observed in Sections \ref{sec:one:d} and
\ref{sec.toy},  it is less affected by particle
degeneracy. Consequently, it is able to more
efficiently adapt the Markov mutation kernels through the adaptation
strategy \eqref{eq.adapt.PCNL}. Figure \ref{fig:IP_temp_traj} (right)
displays the averaged value of the data-misfit functional $(-1/N)
\sum_{j=1}^N V(\uu^N_j)$ as a function of the total number of mutation
steps. As displayed in Figure \ref{fig:IP_STD_MEAN}, the estimation of 
the posterior marginal mean and standard deviation produced by 
the SMC and SET methods are equally good, and agree well with the 
MCMC simulations.

%
%
\begin{figure}[h!t!b!]
\centering
\subfigure{
\includegraphics[width=1\textwidth]{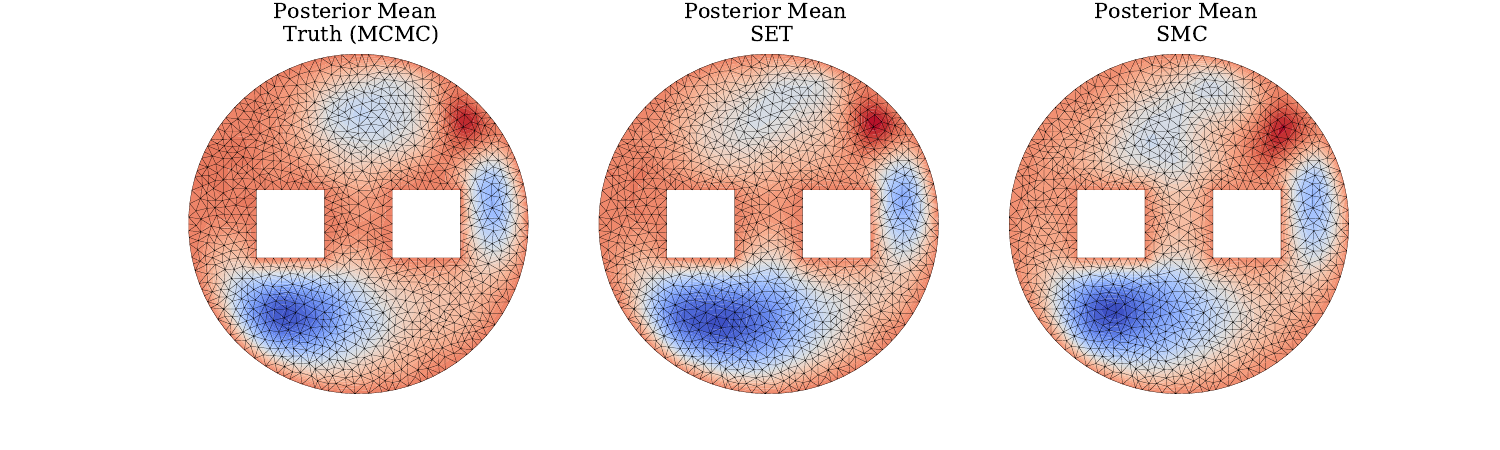}
}
\subfigure{
\includegraphics[width=1\textwidth]{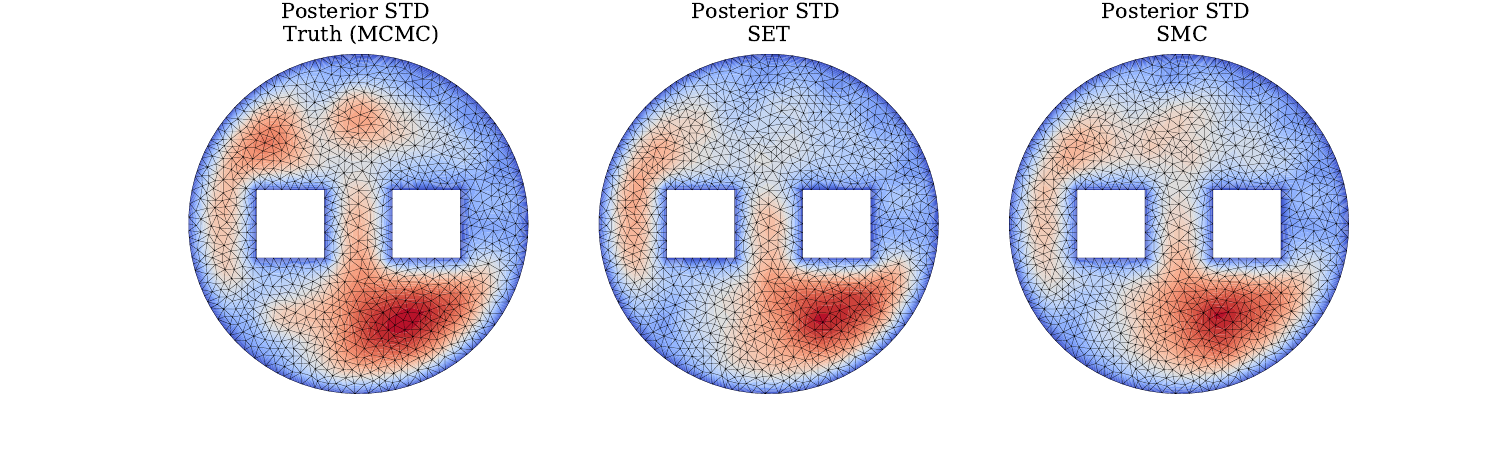}
}
\caption{Bayesian Inverse Problem \eqref{eq.elliptic.pde}. 
{\bf (First row:)} posterior mean obtained from MCMC {\bf (left)} and SET {\bf (center)} and SMC {\bf (right)}. 
{\bf (Second row:)} posterior marginal standard deviations obtained from MCMC {\bf (left)} and SET {\bf (center)} and SMC {\bf (right)}. 
The SET and SMC methods were used with $N=2.10^3$ particles, adaptive PCNL mutation kernels, an adaptive temperature scheme, and an adaptive number of mutation steps at each temperature.
}
\label{fig:IP_STD_MEAN}
\end{figure}

We conclude this section with a brief discussion of effective sample
size computations. Although there have been a few recent and important
methodological advances in this area
\cite{chan2013general,lee2018variance,olsson2019numerically,du2019variance},
it is fair to say that it is still difficult reliably to evaluate the effective
sample size 
for 
interacting particles methods
such as SMC or the SET. It is worth emphasizing that the effective
sample size functional defined in Equation \eqref{ess} is only used
for adapting the temperature ladder: it is not designed, nor should be
used, to provide a reliable estimate of the variability of the
quantities derived from a particle system.

\section{Conclusions}
\label{sec:conclusions}
\seclab{conclusions}
We have introduced the SET method, an optimal-transport based approach for performing inference in high-dimensional Bayesian inverse problems. The SET methodology is, under mild assumptions, provably consistent in the large-particle regime. Our numerical simulations indicate that, in complex high-dimensional scenarios such as PDE-constrained Bayesian inverse problems where it is typically difficult to design efficient Markov mutation kernels, the SET method performs favourably when compared to other particle-based approaches such as modern adaptive SMC methodologies. 
Our numerical results indicate that the SET method, by relying on transportation methods instead of a resampling scheme, is less affected than SMC by particles degeneracy and is able to better exploit the particles system to adapt the mutation kernels. Although our theoretical results provide consistency guarantees, they do not quantify nor explain the empirical gains observed when comparing the SET to standard SMC approaches.

\bigskip
{\bf Acknowledgement:}
 AM and TBT are partially supported by the Department of
 Energy (grant DE-SC0018147), the National Science Foundation (grants
 NSF-DMS1620352, Early Career NSF-OAC1808576, and NSF-OAC1750863), the
 Defense Threat Reduction Agency  (grant HDTRA1-18-1-0020), and a 2018 ConTex award. AM thanks Nick Alger for many useful discussions on Bayesian inverse problems. AHT acknowledges support from a National University of Singapore (NUS) Young Investigator Award Grant (R-155-000-180-133) and a Singapore Ministry of Education Academic Research Funds Tier 2 (MOE2016-T2-2-135).

\bibliographystyle{alpha}
\bibliography{ref3}

\end{document}